\documentclass{jfm}

\usepackage{graphicx}
\usepackage{natbib}

\ifCUPmtlplainloaded \else
  \checkfont{eurm10}
  \iffontfound
    \IfFileExists{upmath.sty}
      {\typeout{^^JFound AMS Euler Roman fonts on the system,
                   using the 'upmath' package.^^J}%
       \usepackage{upmath}}
      {\typeout{^^JFound AMS Euler Roman fonts on the system, but you
                   dont seem to have the}%
       \typeout{'upmath' package installed. JFM.cls can take advantage
                 of these fonts,^^Jif you use 'upmath' package.^^J}%
      }
  \else
  \fi
\fi

\ifCUPmtlplainloaded \else
  \checkfont{msam10}
  \iffontfound
    \IfFileExists{amssymb.sty}
      {\typeout{^^JFound AMS Symbol fonts on the system, using the
                'amssymb' package.^^J}%
       \usepackage{amssymb}%
         \let\leq=\leqslant
         \let\geq=\geqslant
      }{}
  \fi
\fi

\ifCUPmtlplainloaded \else
  \IfFileExists{amsbsy.sty}
    {\typeout{^^JFound the 'amsbsy' package on the system, using it.^^J}%
     \usepackage{amsbsy}}
    {\providecommand\boldsymbol[1]{\mbox{\boldmath $##1$}}}
\fi

\providecommand\bnabla{\boldsymbol{\nabla}}

\newsavebox{\astrutbox}
\sbox{\astrutbox}{\rule[-5pt]{0pt}{20pt}}

\newcommand{\beq}{\begin{equation}}
\newcommand{\eeq}{\end{equation}}

\newcommand{\bfb}{\mbox{\boldmath $b$}}

\newcommand{\bfh}{\mbox{\boldmath $h$}}
\newcommand{\bfk}{\mbox{\boldmath $k$}}
\newcommand{\bfu}{\mbox{\boldmath $u$}}
\newcommand{\bfv}{\mbox{\boldmath $v$}}
\newcommand{\bfx}{\mbox{\boldmath $x$}}
\newcommand{\bfr}{\mbox{\boldmath $r$}}
\newcommand{\bfB}{\mbox{\boldmath $B$}}
\newcommand{\bfH}{\mbox{\boldmath $H$}}
\newcommand{\bfK}{\mbox{\boldmath $K$}}
\newcommand{\bfR}{\mbox{\boldmath $R$}}
\newcommand{\bfX}{\mbox{\boldmath $X$}}

\newcommand{\bfxi}{\mbox{\boldmath $\xi$}}
\newcommand{\ex}{\mbox{{\boldmath $e$}}_{1}}
\newcommand{\ey}{\mbox{{\boldmath $e$}}_{2}}
\newcommand{\ez}{\mbox{{\boldmath $e$}}_{3}}
\newcommand{\bfemf}{\mbox{\boldmath ${\cal E}$}}
\newcommand{\cross}{\mbox{\boldmath $\times$}}
\newcommand{\cendot}{\mbox{\boldmath $\cdot\,$}}
\newcommand{\rem}{{\rm Rm}}

\title[Shear dynamo problem]{The shear dynamo problem for small magnetic
Reynolds numbers}

\author[S. Sridhar and Nishant K. Singh]%
{S.\ns S\ls R\ls I\ls D\ls H\ls A\ls R$^1$
\thanks{Electronic address: ssridhar@rri.res.in}
\and N\ls I\ls S\ls H\ls A\ls N\ls T\ns
K.\ns S\ls I\ls N\ls G\ls H$^{1,2}$
\thanks{Electronic address: nishant@rri.res.in}
}

\affiliation{$^1$Raman Research Institute, Sadashivanagar, Bangalore
 560 080, India\\[\affilskip]
$^2$Joint Astronomy Programme, Indian Institute of Science, Bangalore 560 012, India}

\date{}
\pubyear{} 
\volume{}
\pagerange{\pageref{firstpage}--\pageref{lastpage}}
\doi{S002211200100456X}

\begin{document}

\label{firstpage}
\maketitle

\begin{abstract}
We study large--scale kinematic dynamo action due to turbulence in the presence of a linear 
shear flow, in the low conductivity limit. Our treatment is non perturbative in the 
shear strength and makes systematic use of both the shearing coordinate transformation 
and the Galilean invariance of the linear shear flow. The velocity fluctuations are assumed 
to have low magnetic Reynolds number ($\rem$) but could have arbitrary fluid Reynolds 
number. The equation for the magnetic fluctuations is expanded perturbatively in the small quantity, 
$\rem$. Our principal results are as follows: (i) The magnetic fluctuations are determined to
lowest order in $\rem$ by explicit calculation of the resistive Green's function for the linear 
shear flow; (ii) The mean electromotive force is then calculated and an integro--differential 
equation is derived for the time evolution of the mean magnetic field. In this equation, velocity 
fluctuations contribute to two different kinds of terms, the ``C'' and ``D'' terms, in which first and second 
spatial derivatives of the mean magnetic field, respectively, appear inside the spacetime 
integrals; (iii) The contribution of the ``D'' terms is such that their contribution to the time evolution of the cross--shear 
components of the mean field do not depend on any other components excepting themselves.
Therefore, to lowest order in $\rem$ but to all orders in the shear strength, the ``D'' terms cannot
give rise to  a shear--current assisted dynamo effect;  (iv) Casting the integro--differential equation in Fourier space, 
we show that the normal modes of the theory are a set of shearing waves, labelled by their sheared wavevectors; 
(v) The integral kernels are expressed in terms of the velocity spectrum tensor, which is the fundamental 
dynamical quantity that needs to be specified to complete the integro--differential equation description of the time 
evolution of the mean magnetic field; (vi) The ``C'' terms couple different components of the mean magnetic field, so they 
can, in principle, give rise to a shear--current type effect. We discuss the application to a slowly varying magnetic field, 
where it can be shown that forced non helical velocity dynamics at low fluid Reynolds number does not result in a
shear--current assisted dynamo effect.
\end{abstract}

\begin{keywords}
MHD and Electrohydrodynamics --- Dynamo
\end{keywords}

\tableofcontents

\section{Introduction}

Large--scale magnetic fields in many astrophysical systems, from planets to clusters of galaxies, 
are thought to originate from dynamo action in the electrically conducting fluids in these 
objects. The standard paradigm involves amplification of seed magnetic fields, due to 
non mirror--symmetric (i.e. helical) turbulent flows, through the $\alpha$--effect \citep{Mof78, Par79}.
It is only relatively recently that the role of the mean shear in the turbulent flows is beginnng to be 
appreciated. Dynamo action due to shear and turbulence has received some attention in the astrophysical 
contexts of accretion disks \citep{VB97} and galactic disks \citep{Bla98}. It has also been demonstrated that 
shear, in conjunction with rotating turbulent convection, can drive a large--scale dynamo \citep{KKB08, HP09}. We are interested in the more specific problem of large--scale dynamo action due to  non--helical turbulence with mean shear.  Direct numerical simulations now provide strong support for such a {\it shear dynamo}. \citet{You08a} demonstrated that forced small-scale non--helical turbulence in non--rotating linear shear flows leads to exponential growth of large--scale magnetic fields. These findings were later generalized by 
\citet{You08b} to a shearing sheet model of a differentially rotating disk with a Keplerian rotation profile. The investigations of 
\citet{BRRK08} demonstrated the shear dynamo effect for a range of values of the Reynolds numbers and the shear parameter, and measured the tensorial magnetic diffusivity tensor. While the shear dynamo has been conclusively demonstrated to function, it is not yet clear what makes it work. This outstanding, unsolved problem is the focus of the present investigation.   

One possibility that has been suggested is dynamo action due to a ``fluctuating $\alpha$--effect'' in 
turbulent flows which have zero mean helicities. In this proposal, large--scale dynamo action derives from the 
interaction of mean shear with fluctuations of helicity \citep{VB97,Sok97,Pro07,BRRK08,RK08,Sch08}.
Another suggestion is that, if even transient growth makes non axisymmetric mean magnetic fields strong enough, 
they themselves might drive motions which could lead to subcritical dynamo action \citep{ROPC08}. 
Yet another possibility that has been suggested is the shear--current effect \citep{RK03,RK04,RK08}. 
In this mechanism, it is thought that the mean shear gives rise to anisotropic turbulence, which causes
an extra component of the mean electromotive force (EMF), leading to the generation of the cross--shear 
component of the mean magnetic field from the component parallel to the shear flow. However, there is no agreement 
yet whether the sign of such a coupling is favourable to the operation of a dynamo. Some analytic calculations 
\citep{RS06,RK06} and numerical experiments \citep{BRRK08} find that the sign of the shear--current 
term is unfavourable for dynamo action. A quasilinear theory of dynamo action in a linear shear flow of an 
incompressible fluid which has random velocity fluctuations was presented in \cite{SS09a,SS09b}. Unlike earlier 
analytic work which treated shear as a small perturbation, this work did not place any restriction on the 
strength of the shear. They arrived at an integro--differential equation for the evolution of the mean magnetic 
field and argued that the shear--current assisted dynamo is essentially absent. It should be noted that the 
quasilinear theory of \cite{SS09a,SS09b} assumes zero resistivity, and is valid in the limit of small velocity
correlation times when the  ``first order smoothing approximation'' (FOSA) holds. 

In this paper we present a kinematic theory of the shear dynamo that is non perturbative in the shear strength, but 
perturbative in the magnetic Reynolds number ($\rem$); this may be thought of as FOSA with finite resistivity.
Thus we are not limited to the quasilinear limit of small velocity correlation times, 
and our conclusions are rigorously valid for velocity fluctuations which have small $\rem$ but arbitrary fluid Reynolds number. 
In \textsection~2 we formulate the shear dynamo problem for small $\rem$. Using Reynolds averaging, we split the magnetic field 
into mean and fluctuating components. The equation for the fluctuations is expanded perturbatively in the small parameter, $\rem$.
Using the shearing coordinate transformation, we make an explicit calculation of the resistive Green's 
function for the linear shear flow. In \textsection~3, the magnetic fluctuations and the mean electromotive force 
(EMF) are determined to lowest order in $\rem$.  The transport coefficients are 
given in general form in terms of the two--point correlators of the velocity fluctuations. Galilean invariance is a basic symmetry in 
the problem and is the focus of \textsection~4. For Galilean invariant (G--invariant) velocity fluctuations, it is proved that the transport coefficients, although space-dependent, possess  the property of 
translational invariance in sheared coordinate space. An explicit expression for the Galilean--invariant mean 
EMF is derived. We put together all the results in \textsection~5 by deriving the integro--differential equation 
governing the time evolution of the mean magnetic field. Some important properties of this equation are discussed. 
In particular, it is shown that, in the formal limit of zero resistivity, the quasilinear results of \cite{SS09a,SS09b} are 
recovered.  We also show that the natural setting for the integro--differential equation governing mean--field 
evolution is in sheared Fourier space. We prove a theorem on the form of the two--point velocity correlator
in Fourier space; the velocity spectrum tensor and its general properties are discussed. We then express all the integral kernels in 
terms of the velocity spectrum tensor, which is the fundamental dynamical quantity that needs to be specified. Summary and conclusions are 
presented in \textsection~6.

\section{The shear dynamo problem}

\subsection{The small $\rem$ limit}

Consider a Cartesian coordinate system with unit vectors $(\ex,\ey,\ez)$ erected on a comoving patch of a
differentially rotating disk. Henceforth this will be referred to as the lab frame and we will use
notation $\bfX = (X_1,X_2,X_3)$ for the position vector, and $\tau$ for time. The fluid velocity is given by 
$(SX_1\ey + \bfv)$, where $S$ is the rate of shear parameter and $\bfv(\bfX, \tau)$ is a randomly 
fluctuating velocity field. The total magnetic field, $\bfB^{\rm tot}(\bfX, \tau)$, obeys the induction equation. 

\beq
\left(\frac{\partial}{\partial\tau} \;+\; SX_1\frac{\partial}{\partial X_2}\right)\bfB^{\rm tot} \;-\; SB^{\rm tot}_1\ey \;=\; 
\bnabla\cross\left(\bfv\cross\bfB^{\rm tot}\right) \;+\; \eta\bnabla^2\bfB^{\rm tot}\\[1em]
\label{indeqn}
\eeq
 
\noindent
It is useful to note that the induction equation is unaffected by a 
uniform rotation of the frame of reference. So our coordinate system 
can refer to an inertial frame, or to a comoving patch of a differentially rotating disk. We study a kinematic problem in this paper, so will assume that the velocity field is prescribed. We also assume that the velocity fluctuations have zero mean ($\left<\bfv\right> = {\bf 0}$), with root--mean--squared amplitude $v_{{\rm rms}}$ on some typical spatial scale $\ell$. The 
{\it magnetic Reynolds number} may be defined as $\rem = (v_{{\rm rms}}\ell/\eta)$; note that $\rem$ has been 
defined with respect to the fluctuation velocity field, not the background shear velocity field.
To address the dynamo problem, we will use the approach of the theory of mean--field electrodynamics \citep{Mof78,KR80, BS05}. Here, 
the action of the velocity fluctuations on some seed magnetic field is assumed to 
produce a total magnetic field with a well--defined {\it mean--field} $(\bfB)$ and a {\it fluctuating--field} $(\bfb)$:

\beq
\bfB^{\rm tot} \;=\; \bfB \;+\; \bfb\,,\qquad \left<\bfB^{\rm tot}\right> \;=\; \bfB\,,\qquad \left<\bfb\right> \;=\; {\bf 0}
\label{reynolds}
\eeq

\noindent 
where $\left<\;\;\right>$ denotes ensemble averaging in the sense of Reynolds. Applying Reynolds averaging to the induction equation~(\ref{indeqn}), we obtain the following equations governing the dynamics of the mean and fluctuating magnetic fields:
\begin{eqnarray}
\left(\frac{\partial}{\partial\tau} \;+\; SX_1\frac{\partial}{\partial X_2}\right)\bfB \;-\; SB_1\ey &\;=\;& 
\bnabla\cross\bfemf \;+\; \eta\bnabla^2\bfB\label{meanindeqn}\\[2em] 
\left(\frac{\partial}{\partial\tau} \;+\; SX_1\frac{\partial}{\partial X_2}\right)\bfb \;-\; Sb_1\ey &\;=\;& 
\bnabla\cross\left(\bfv\cross\bfB\right) \;+\; \bnabla\cross\left(\bfv\cross\bfb - \left<\bfv\cross\bfb\right>\right) \;+\; \eta\bnabla^2\bfb\nonumber\\[1ex]
&&\label{flucindeqn}
\end{eqnarray}
\noindent
where $\bfemf = \left<\bfv\cross\bfb\right>$ is the mean electromotive force (EMF). The first step
toward solving the problem is to solve equation~(\ref{flucindeqn}) for $\bfb$, then calculate $\bfemf$ and obtain a closed equation for the mean--field, $\bfB(\bfX, \tau)$. In the framework of the above mean--field theory, the {\it shear dynamo problem} may be posed as follows: under what conditions does the equation for $\bfB(\bfX, \tau)$ admit growing solutions~? In particular, are growing solutions possible when the 
velocity field is non--helical (i.e. when the velocity field is mirror symmetric)~?

The problem is, in general, a difficult one, but it can be approached perturbatively in the limit of small $\rem$. When $\rem\ll 1$, we can expand $\bfb$ in a series, 
\beq
\bfb \;=\; \bfb^{(0)} \;+\; \bfb^{(1)} \;+\; \bfb^{(2)} \;+\; \ldots
\label{bseries}
\eeq
\noindent where $\bfb^{(n)}$ is of order $\bfb^{(n-1)}$ multiplied by the small quantity $\rem$. 
The equations governing the time evolution of these quantitites are
\begin{eqnarray} 
\left(\frac{\partial}{\partial\tau} \;+\; SX_1\frac{\partial}{\partial X_2}\right)\bfb^{(0)} \;-\; Sb_1^{(0)}\ey &=& \bnabla\cross\left(\bfv\cross\bfB\right) \;+\; \eta\bnabla^2\bfb^{(0)}
\label{eqnhier0}\\[2em]
\left(\frac{\partial}{\partial\tau} \;+\; SX_1\frac{\partial}{\partial X_2}\right)\bfb^{(n)} \;-\; Sb_1^{(n)}\ey &=& \bnabla\cross\left(\bfv\cross\bfb^{(n-1)} - \left<\bfv\cross\bfb^{(n-1)}\right>\right) \;+\; \eta\bnabla^2\bfb^{(n)}\nonumber\\[1ex]
&& \qquad\mbox{for $n \;=\; 1,2,\ldots$}
\label{eqnhier}
\end{eqnarray}
\noindent
Note that $\bnabla\cross\left(\bfv\cross\bfB\right)$ acts as a source term for $\bfb^{(0)}$, whereas the source term for $\bfb^{(n)}$ is
 $\bnabla\cross\left(\bfv\cross\bfb^{(n-1)} - \left<\bfv\cross\bfb^{(n-1)}\right>\right)$. Once the $\bfb^{(n)}$ have been determined, the mean 
EMF can be calculated directly by 
\beq
\bfemf \;=\; \left<\bfv\cross\bfb\right> \;=\; \left<\bfv\cross\left(\bfb^{(0)} \;+\; \bfb^{(1)} \;+\; \bfb^{(2)} \;+\; \ldots\right)\right> 
\label{emfseries}
\eeq

\noindent
In this paper, we work to lowest order in $\rem$, so we need to work out only $\bfb^{(0)}$; equation~(\ref{eqnhier}) will not be used.

\subsection{The shearing coordinate transformation}

In this paper we will focus on the determination of the lowest order term, $\bfb^{(0)}$. We also assume that the fluctuating velocity field is incompressible; i.e. $\bnabla\cendot\bfv \;=\; 0$. Then the evolution of $\bfb^{(0)}$ is governed by,
\beq 
\left(\frac{\partial}{\partial\tau} \;+\; SX_1\frac{\partial}{\partial X_2}\right)\bfb^{(0)} \;-\; Sb_1^{(0)}\ey \;=\; \left(\bfB\cendot\bnabla\right)\bfv \;-\; \left(\bfv\cendot\bnabla\right)\bfB \;+\; \eta\bnabla^2\bfb^{(0)}
\label{flucindlineqn}
\eeq
\noindent
We will now solve this equation for $\bfb^{(0)}$ and determine the mean EMF. General methods of solving
equations such as equation~(\ref{flucindlineqn}) are presented in \cite{KR80}, but we prefer to
employ the shearing coordinate transformation because it is directly adapted to the problem at hand and  
greatly simplifies the task of writing down the Green's function solution. The $\left(X_1{\partial/\partial X_2}\right)$ term makes equation~(\ref{flucindlineqn}) inhomogeneous in the coordinate $X_1$. This term can be eliminated through a shearing transformation to new spacetime variables:
\beq
x_1 = X_1\,,\qquad x_2 = X_2 - S\tau X_1\,,\qquad x_3 = X_3\,,\qquad t = \tau
\label{sheartr}
\eeq
\noindent
Partial derivatives transform as
\beq
\frac{\partial}{\partial\tau} = \frac{\partial}{\partial t} - 
Sx_1\frac{\partial}{\partial x_2}\,,\quad 
\frac{\partial}{\partial X_1} = \frac{\partial}{\partial x_1} - 
St\frac{\partial}{\partial x_2}\,,\quad 
\frac{\partial}{\partial X_2}
= \frac{\partial}{\partial x_2}\,,\quad
\frac{\partial}{\partial X_3}
= \frac{\partial}{\partial x_3}
\label{partialsh}
\eeq
\noindent 
Therefore
\begin{eqnarray}
\bnabla^2 \equiv \frac{\partial^2}{\partial X_p \partial X_p} &=& \left( \frac{\partial}{\partial x_p} - St\delta _{p1} \frac{\partial}{\partial x_2} \right)^2 \nonumber\\[2ex]
&=& \frac{\partial^2}{\partial x_p \partial x_p} - 2St \frac{\partial ^2}{\partial x_1 \partial x_2} + S^2t^2\frac{\partial^2}{\partial x^2_2}
\label{delsq} 
\end{eqnarray}
\noindent
We also define new variables, which are component--wise equal to the old variables: 
\beq
\bfH(\bfx, t) \;=\; \bfB(\bfX, \tau)\,,\qquad
\bfh(\bfx, t) \;=\; \bfb^{(0)}(\bfX, \tau)\,,\qquad
\bfu(\bfx, t) \;=\; \bfv(\bfX, \tau)
\label{newvar}
\eeq
\noindent
Note that, just like the old variables, the new variables are expanded in the fixed Cartesian basis of the lab frame. For example, $\bfH = H_1\ex + H_2\ey + H_3\ez$, where $H_i(\bfx, t) = B_i(\bfX, \tau)$, and similarly for the other variables. In the new variables, equation~(\ref{flucindlineqn}) becomes,
\beq
\frac{\partial\bfh}{\partial t} \;-\; Sh_1\ey \;=\; \left(\bfH\cendot\frac{\partial}{\partial\bfx}
- StH_1\frac{\partial}{\partial x_2}\right)\bfu \;-\; \left(\bfu\cendot\frac{\partial}{\partial\bfx} - Stu_1\frac{\partial}{\partial x_2}\right)\bfH \;+\; \eta\bnabla^2\bfh
\label{newvareqn}
\eeq
\noindent
which can be expressed in component form as
\beq
\left(\frac{\partial}{\partial t} \;-\; \eta \bnabla^2\right)h_m(\bfx,t) \;=\; q_m(\bfx,t) 
\label{indflsh}
\eeq
\noindent
where $\bnabla^2$ is given by equation~(\ref{delsq}), and 
\beq
q_m(\bfx,t) \;=\; \left[ H_l - St \delta_{l2} H_1 \right]u_{ml} \;-\;
\left[ u_l - St \delta_{l2} u_1 \right]H_{ml} \;+\; S \delta_{m2} h_1 
\label{sourceq} 
\eeq
\noindent
We have used notation $u_{ml} =(\partial u_m/\partial x_l)$ and $H_{ml} = (\partial H_m/\partial x_l)$. 
Below we construct the Green's function for equation~(\ref{indflsh}). 

\subsection{The resistive Green's function for a linear shear flow}

Equation~(\ref{indflsh}) is linear, homogeneous in $\bfx$ and inhomogeneous in $t$. Therefore, the general solution can be written in the form,
\begin{eqnarray}
h_m(\bfx,t) \;=\; && \int \mathrm{d}^3x'\,G_{\eta}(\bfx-\bfx',t,s)\;h_m(\bfx',s)\nonumber\\[3ex]
&+& \int_s^t \mathrm{d}t'\int \mathrm{d}^3x'\,G_{\eta}(\bfx-\bfx',t,t')\;q_m(\bfx',t')\;;\qquad\mbox{for any $s < t\,$,} 
\label{grfnsoln}
\end{eqnarray}
\noindent
where $G_{\eta}(\bfx,t,t')$ is the {\it resistive Green's function} for the linear shear flow, which satisfies, 

\begin{subequations}
\beq
\left(\frac{\partial}{\partial t} - \eta \bnabla^2 \right)G_{\eta}(\bfx,t,t') \;=\; 0  
\label{grfnpr1}
\eeq
\beq
\lim_{t'\to t_-}\;G_{\eta}(\bfx,t,t') \;=\; \delta^3 (\bfx)  
\label{grfnpr2}
\eeq
\beq
G_{\eta}(\bfx,t,t')\quad\mbox{is non--zero only when}\quad 0\leq t' < t. 
\label{grfnpr3}
\eeq
\beq
G_{\eta}(\bfx-\bfx',t,t_0) \;=\; \int \mathrm{d}^3x^{''}\,G_{\eta}(\bfx-\bfx^{''},t,s)\,G_{\eta}(\bfx^{''}-\bfx',s,t_0)
\,;\quad\mbox{for $t_0 < s < t\,$.}
\label{grfnpr4}
\eeq
\end{subequations}

Let us define the spatial Fourier transform of the Green's function as,
\beq
\widetilde{G}_{\eta}(\bfk,t,t') \;=\; \int \mathrm{d}^3x \, \exp{(-\mathrm{i}\,\bfk\cendot\bfx)} G_{\eta}(\bfx,t,t')
\label{frgr}
\eeq
\noindent
where $\bfk$, being conjugate to the sheared coordinate vector $\bfx$, can be regarded as a 
{\it sheared wavevector}. Then 

\beq
\frac{\partial\,\widetilde{G}_{\eta}}{\partial t} \;+\; \eta \, K^2(\bfk,t) \, \widetilde{G}_{\eta} \; = \; 0
\label{frgr3}
\eeq

\noindent
where, in equation~(\ref{frgr3}), $K^2(\bfk,t)=(k_1 - Stk_2)^2 + k_2^2 + k_3^2$. It is now straightforward to write down the solution:
\begin{eqnarray}
\widetilde{G}_{\eta}(\bfk,t,t') &=& \exp{\left[-\eta \int_{t'}^t \mathrm{d}s \, K^2(\bfk,s)\right]}\nonumber\\[3ex]
&=& \exp{\left[-\eta\left(k^2(t-t') - S\,k_1\,k_2(t^2-t^{\prime 2}) + \frac{S^2}{3}\,k_2^2 (t^3-t^{\prime 3}) \right)\right]}
\label{frgrsoln}
\end{eqnarray}
\noindent
where, as per equation~(\ref{grfnpr3}) above, $t > t'\,$.
Note also that $\widetilde{G}_{\eta}(\bfk,t,t')$ is a positive quantity which takes values between $0$ and $1$, and that it is an even function of $\bfk$ and $k_3$.

We now take the inverse Fourier transform of equation~(\ref{frgrsoln}) to get $G_{\eta}(\bfx,t,t')$. It is 
convenient to write this as
\beq
G_{\eta}(\bfx,t,t') \,=\, \int \frac{\mathrm{d}^3k}{(2 \pi)^3} \exp{\left [\mathrm{i}\,\bfk\cendot\bfx \, - \eta(t-t') \{ k^2 + T_{ij} \,k_i \,k_j \} \right ]}
\label{invfrtr}
\eeq
\noindent where $T_{ij}$ is a $2\times 2$ symmetric matrix whose elements are given by,
\beq
T_{11} \;=\; 0,\quad T_{12} \;=\; T_{21} \;=\; -\frac{S}{2}(t+t'),\quad T_{22} \;=\; 
\frac{S^2}{3}(t^2 + t t' + t^{\prime 2})
\label{matrixq}
\eeq
\noindent
The integral in equation~(\ref{invfrtr}) can be evaluated by diagonalising the matrix
$T_{ij}$. It proves useful to express $G_{\eta}(\bfx,t,t')$ in terms of the principal--axes coordiates, 
$\overline{\bfx} = \left(\overline{x}_1, \,\overline{x}_2, \,\overline{x}_3\right)$. These are defined by
the orthogonal transformation, 
\beq
\left(
\begin{array}{c}
\overline{x}_1\\[3ex]
\overline{x}_2\\[3ex]
\overline{x}_3
\end{array}
\right) \;=\;
\left(
\begin{array}{ccc}
\cos\theta & \qquad  \sin\theta & \qquad 0\\[3ex]
-\sin\theta & \qquad \cos\theta & \qquad 0\\[3ex]
0 & \qquad 0 & \qquad 1
\end{array}
\right)
\left(
\begin{array}{c}
x_1\\[3ex]
x_2\\[3ex]
x_3
\end{array}
\right)
\eeq
\noindent which is a time--dependent rotation of the coordinate axes in the $x_1$--$x_2$ plane. The angle of 
rotation, $\theta$, is determined by
\begin{eqnarray}
\tan\theta &=& f \;+\; \sqrt{1+f^2}\nonumber\\[3ex]
f &=&  -\frac{1}{3}\,\frac{S(t^2 + t t' + t^{\prime 2})}{(t+t')}
\label{thetadef}
\end{eqnarray}
\noindent
Note that $\theta$ depends on the shear parameter, $S$, and the times, $t$ and $t'$. 
Let us define the dimensionless quantitites,
\begin{eqnarray}
\sigma_1 &=& \left[1 \;-\; \frac{S}{2}(t + t')\,\tan\theta\right]^{1/2}\nonumber\\[3ex]
\sigma_2 &=& \left[1 \;+\; \frac{S}{2}(t + t')\,\cot\theta\right]^{1/2}
\label{a12def}
\end{eqnarray}

\begin{figure}
\includegraphics[scale=0.55, angle=-90]{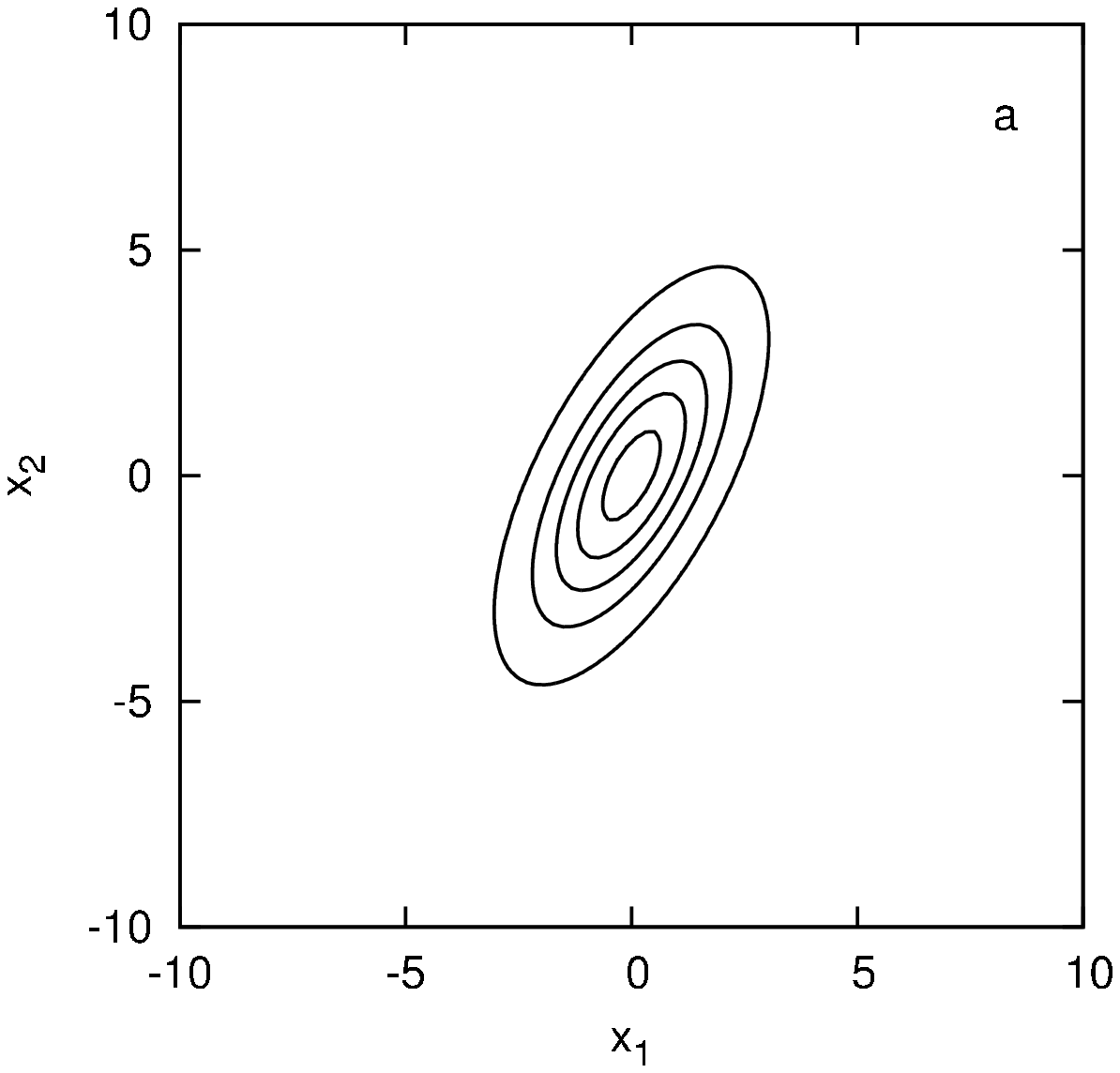}
\hskip0.2cm
\includegraphics[scale=0.55, angle=-90]{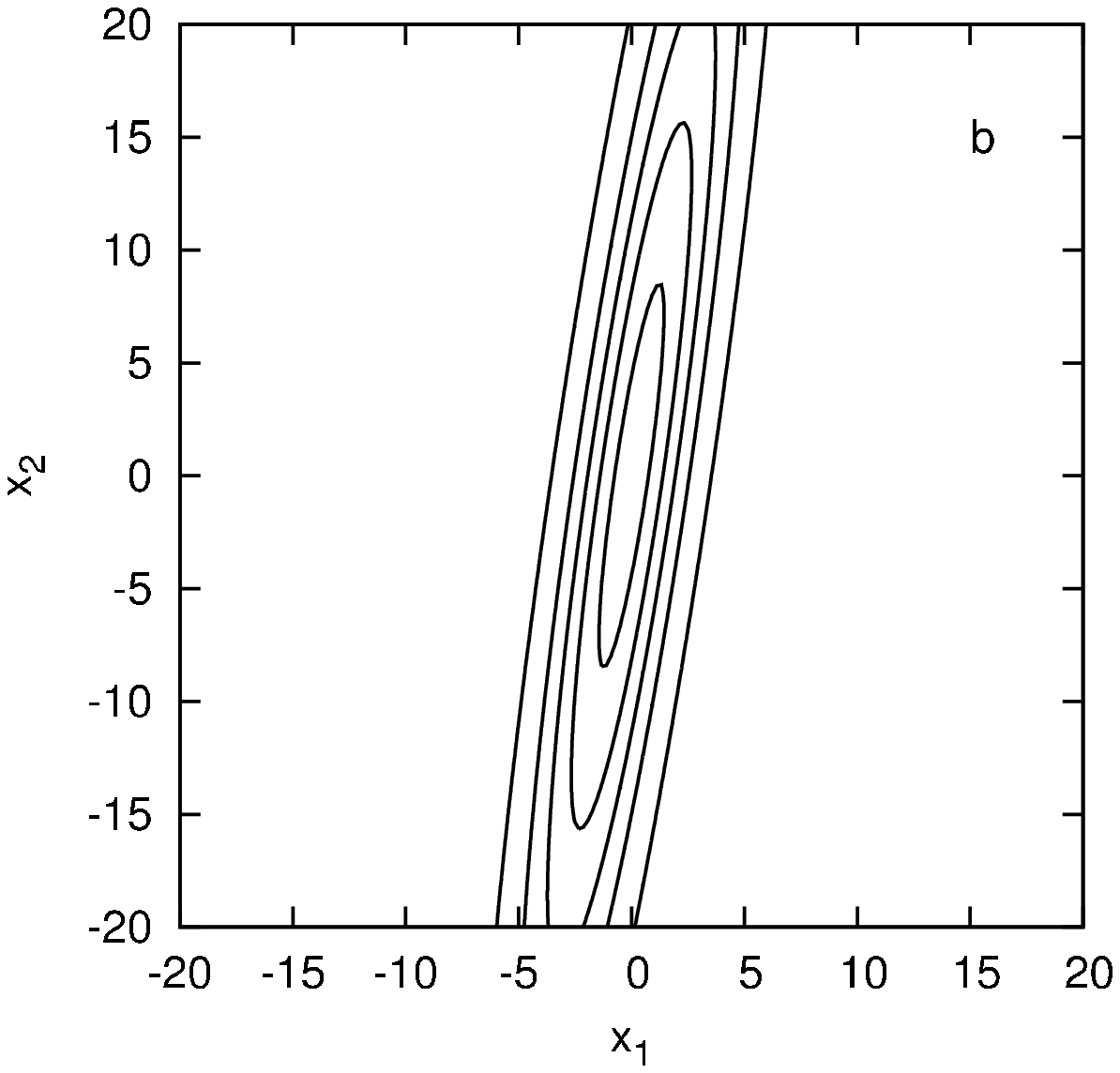}
\vskip0.2cm 
\includegraphics[scale=0.55,angle=-90]{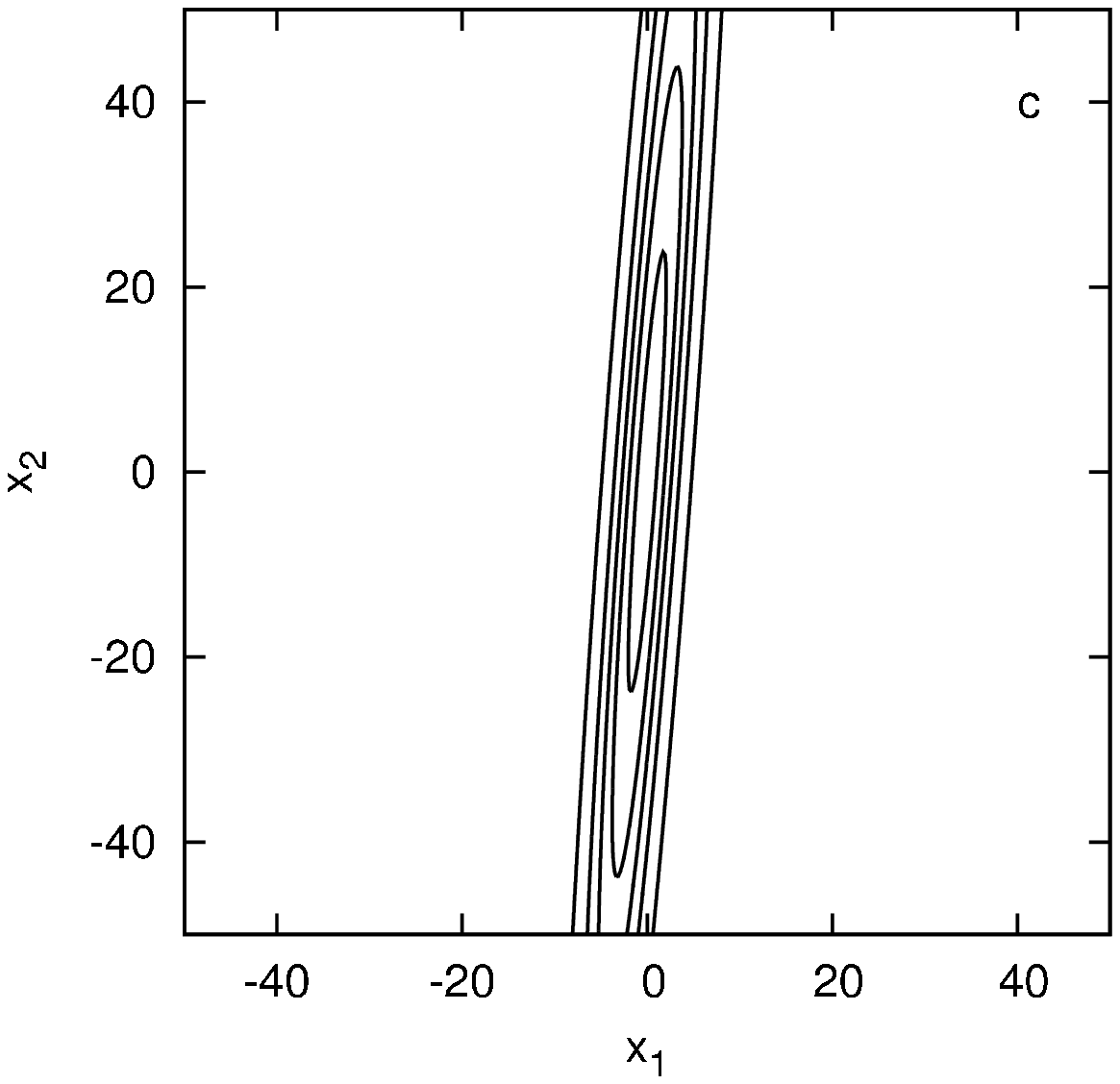}
\hskip0.2cm
\includegraphics[scale=0.55,angle=-90]{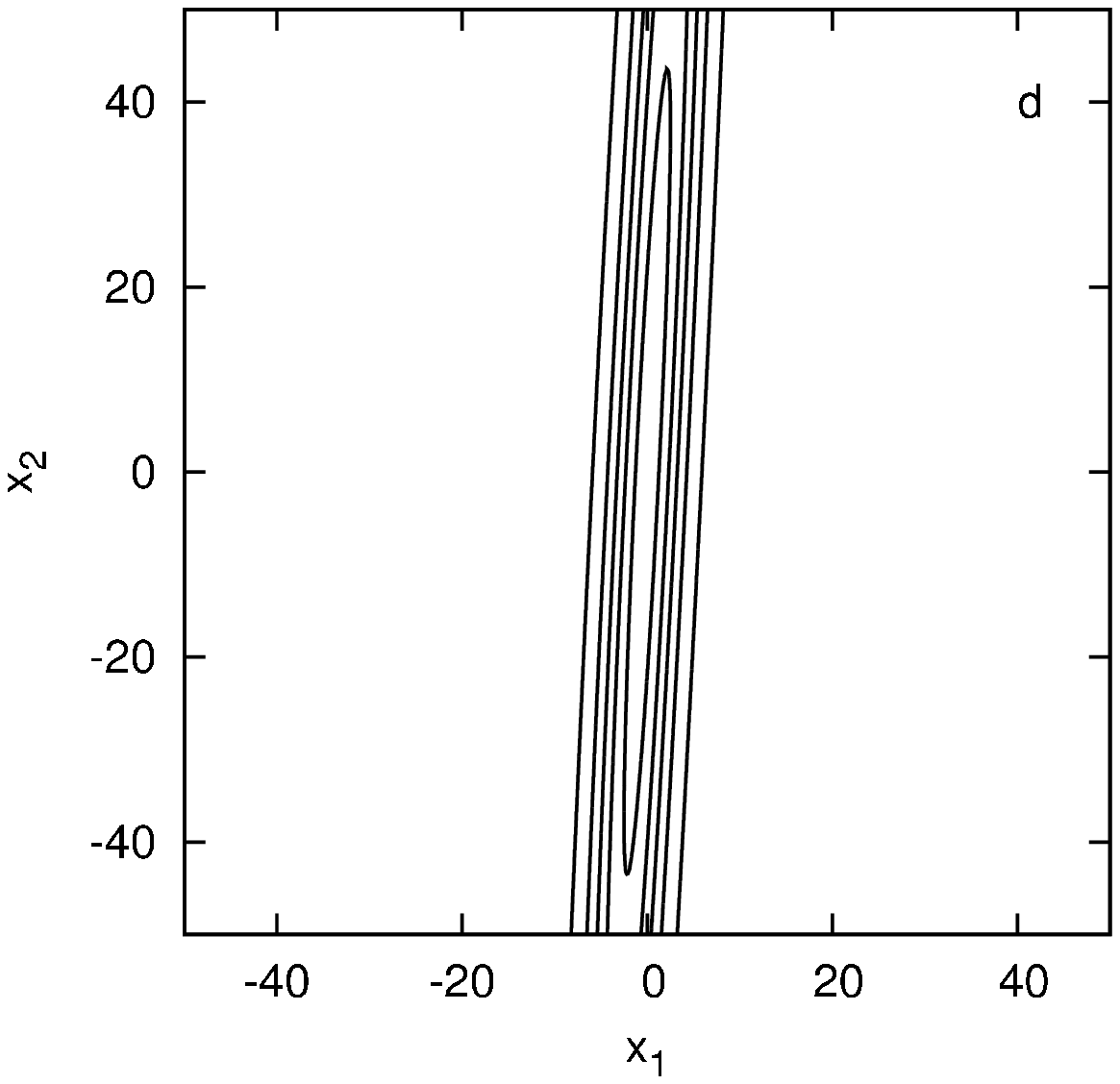}
\caption{Isocontours of the resistive Green's function $G_{\eta}(\bfx,t,t')$ plotted in 
the $x_1$--$x_2$ plane of the shearing coordinate system, for $t'=0$ at four different values of $t$. 
Units are such that $S=-2\,;\eta=1$. Five isocontours at $90\%$, $70\%$, $50\%$, $30\%$ and $10\%$ of the maximum value are displayed.
Panels (a), (b), (c) and (d) correspond to times $t=1$, $t=5$, $t=10$ and $t=15$.}
\label{green_fn}
\end{figure}

\noindent
Now we can write the Green's function as a {\it sheared heat kernel},
\begin{eqnarray}
G_{\eta}(\bfx,t,t')\,=\, && \left[ 4 \pi \eta(t-t') \right]^{-3/2} \left[ 1 \,+\, \frac{S^2}{12} (t - t')^2 \right]^{-1/2} \times \nonumber \\ [2ex]
&\times& \exp{ \left[-\frac{1}{4 \eta(t-t')} \left( \frac{\overline{x}_1^2}{\sigma_1^2} \,+\, \frac{\overline{x}_2^2}{\sigma_2^2} 
\,+\, \overline{x}_3^2 \right) \right] }\,,
\end{eqnarray}
\noindent
which is equivalent to the one first derived in \cite{KR71}. 

We now note some properties of the Green's function. For convenience we choose the shear parameter, $S$, to be negative: 
then  the quantities, $f \geq 0$, $0\leq\theta\leq \pi/2$,  $\sigma_1 \geq\ 1$ and $0\leq \sigma_2\leq 1\,$.
At fixed $t$ and $t'$, the Green's function is a Gaussian with {\it long axis} along 
$\overline{x}_1$, {\it short axis} along $\overline{x}_2$, and the {\it intermediate axis}
along $\overline{x}_3$. To obtain some idea of the behaviour of the Green's function, it is useful to plot 
isocontours in the sheared coordinate space $(x_1, x_2, x_3)$ at different values of $t$ and $t'$. 
Figure~(\ref{green_fn}) displays isocontours in the $x_1$--$x_2$ plane at four different values of 
$t$ for $t'=0$; we have chosen $x_3=0$ and $t'=0$ in the interests of brevity of presentation. The figure is plotted in shearing coordinates, with respect to which diffusion is anisotropic
and there is no advection. It may be noted that the Green's function shows  a shearing motion \emph{against} the direction of the actual shear.
As $t$ increases from zero to infinity, $\theta$ (which is the angle the long axis makes with the $x_1$--axis)
increases from $45^{\circ}$ to $90^{\circ}$, and all the principal axes increase without bound.

\section{Magnetic fluctuations and mean EMF at small $\rem$}

\subsection{Explicit solution for $\bfh(\bfx, t)$}

We are interested in the particular solution to equation~(\ref{indflsh}) (i.e. the {\it forced solution}) which vanishes at $t=0$. This can be written as
\beq
h_m(\bfx,t) \;=\; \int_0^t \mathrm{d}t' \int \mathrm{d}^3x' \, G_{\eta}(\bfx-\bfx',t,t') \; q_m(\bfx',t')
\label{solgrfn}
\eeq
\noindent
Substituting the expression for $q_m$ from equation~(\ref{sourceq}) in equation~(\ref{solgrfn}), 
we have
\begin{eqnarray}
h_m(\bfx,t) &=& \int_0^t \mathrm{d}t' \int \mathrm{d}^3x' \, G_{\eta}(\bfx-\bfx',t,t') \times \nonumber\\[2ex]
&&\qquad\qquad\times \left\{\left[ H'_l - St'\delta_{l2} H'_1 \right]u'_{ml} \;-\; \left[u'_l - St'\delta_{l2}u'_1\right]H'_{ml}\right\}\nonumber\\[3ex]
&& \;+\; S \delta_{m2}\int_0^t \mathrm{d}t'\int \mathrm{d}^3x' \, G_{\eta}(\bfx-\bfx',t,t') \; h_1(\bfx', t')
\label{h_msoln1}
\end{eqnarray}
\noindent
where primes denote evaluation at spacetime point $(\bfx', t')$. The solution is not yet
in explicit form because the last term on the right side contains the unknown quantity $h_1(\bfx', t')$.
Thus we need to work out the integral
\begin{eqnarray}
\int_0^t \mathrm{d}t'\int \mathrm{d}^3x'\,G_{\eta}(\bfx-\bfx',t,t')\,h_1(\bfx', t') &=& \int_0^t \mathrm{d}t'\int \mathrm{d}^3x'\, G_{\eta}(\bfx-\bfx',t,t')\times\nonumber\\[2ex]
&&\times\int_0^{t'} \mathrm{d}t^{''} \int \mathrm{d}^3x^{''}\, G_{\eta}(\bfx'-\bfx^{''},t',t^{''})\times\nonumber\\[2ex]
&\times& \left\{\left[ H_l^{''} - St^{''} \delta_{l2} H_1^{''} \right]u_{1l}^{''} - \left[u_l^{''} - St^{''} \delta_{l2} u_1^{''}\right]H_{1l}^{''}\right\}\nonumber
\end{eqnarray}
\noindent
where $''$ means evaluation at spacetime point $(\bfx^{''},t^{''})$. Note that, on the right side, $\bfx'$ occurs only in the Green's functions. So, by using the 
property given in equation~(\ref{grfnpr4}), the integral over $\bfx'$ can be performed. Then 
\begin{eqnarray}
\int_0^t \mathrm{d}t'\int \mathrm{d}^3x'\,G_{\eta}(\bfx-\bfx',t,t')\,h_1(\bfx', t') &=& \int_0^t \mathrm{d}t'
\int_0^{t'} \mathrm{d}t^{''}\int \mathrm{d}^3x^{''}\, G_{\eta}(\bfx-\bfx^{''},t,t^{''})\times\nonumber\\[2ex]
&\times& \left\{\left[ H_l^{''} - St^{''} \delta_{l2} H_1^{''} \right]u_{1l}^{''} - \left[u_l^{''} - St^{''} \delta_{l2} u_1^{''}\right]H_{1l}^{''}\right\}\nonumber
\end{eqnarray}
\noindent
The double--time integrals can be reduced to a single--time integrals because of the
following simple identity. For any function $f(\bfx, t)$, we have
\begin{eqnarray}
\int_0^t \mathrm{d}t'\int_0^{t'} \mathrm{d}t^{''} \int \mathrm{d}^3x^{''}\,f(\bfx^{''}, t^{''}) &=& 
\int_0^t \mathrm{d}t^{''} \int \mathrm{d}^3x^{''}\,f(\bfx^{''}, t^{''})\int_{t^{''}}^{t} \mathrm{d}t' \nonumber \\ [2ex]
&=&\int_0^t \mathrm{d}t^{''}\,(t - t^{''}) \int \mathrm{d}^3x^{''}\,f(\bfx^{''}, t^{''})\nonumber\\[2ex]
&=& \int_0^t \mathrm{d}t'\,(t-t') \int \mathrm{d}^3x'\,f(\bfx', t')\nonumber
\end{eqnarray}
\noindent
where in the last equality we have merely replaced the dummy integration variables $(\bfx^{''},t^{''})$ by $(\bfx',t')$. Then we have
\begin{eqnarray}
\int_0^t \mathrm{d}t'\int \mathrm{d}^3x'\,G_{\eta}(\bfx-\bfx',t,t')\; h_1(\bfx', t') &=& 
\int_0^t \mathrm{d}t'\,(t-t') \int \mathrm{d}^3x'\,G_{\eta}(\bfx-\bfx^{''},t,t^{''})\times\nonumber\\[2ex]
&\times&\left\{\left[ H'_l - St' \delta_{l2} H'_1 \right]u'_{1l} - \left[u'_l - St' \delta_{l2} u'_1\right]H'_{1l}\right\}\nonumber
\end{eqnarray}
\noindent
Therefore the forced solution to equation~(\ref{indflsh}) can finally be written in explicit form as
\begin{eqnarray}
h_m(\bfx,t) &=& \int_0^t \mathrm{d}t' \int \mathrm{d}^3x' \, G_{\eta}(\bfx-\bfx',t,t')\,
\left[u'_{ml} + S(t-t')\delta_{m2} u'_{1l}\right] \times \nonumber \\[2ex]
&& \qquad\qquad\qquad\qquad\qquad \times \left[H'_l - St'\delta_{l2}H'_1\right]\nonumber\\[3ex]
&-& \int_0^t \mathrm{d}t'\int \mathrm{d}^3x'\,G_{\eta}(\bfx-\bfx',t,t')\,
\left[H'_{ml} + S(t-t')\delta_{m2}H'_{1l}\right] \times \nonumber \\[2ex]
&& \qquad\qquad\qquad\qquad\qquad \times \left[u'_l - St'\delta_{l2}u'_1\right]
\label{hsoln}
\end{eqnarray}
\noindent
This gives the magnetic fluctuation to lowest order in $\rem$.

\subsection{Explicit expression for the mean EMF}

To lowest order in $\rem$, the mean EMF is given by $\bfemf = \left<\bfv\cross\bfb^{(0)}\right> = \left<\bfu\cross\bfh\right>$, where equation~(\ref{hsoln}) for $\bfh$ should be substituted. 
The averaging, $\left<\;\;\right>$, acts only on the velocity variables but not the mean field; i.e.
$\left<\bfu\bfu\bfH\right> = \left<\bfu\bfu\right>\bfH$ etc. After interchanging the dummy indices 
$(l,m)$ in the last term, the mean EMF is given in component form as 
\begin{eqnarray}
{\cal E}_i(\bfx, t) \;=\;&& \epsilon_{ijm}\left<u_jh_m\right>\nonumber\\[3ex]
\;=\;&& \int_0^t \mathrm{d}t' \int \mathrm{d}^3x'\,G_{\eta}(\bfx-\bfx',t,t') \left[\widehat{\alpha}_{il}(\bfx, t, \bfx', t') \;+\; S(t-t')\widehat{\beta}_{il}(\bfx, t, \bfx', t') \right] \times \nonumber \\ [2ex]
&&\qquad\qquad\times \left[H'_l \;-\; St'\delta_{l2}\,H'_1\right]\nonumber\\[3ex]
&-&\int_0^t \mathrm{d}t' \int \mathrm{d}^3x'\,G_{\eta}(\bfx-\bfx',t,t') \left[\,\widehat{\eta}_{iml}(\bfx, t, \bfx', t') \;-\; St'\delta_{m2}\,\widehat{\eta}_{i1l}(\bfx, t, \bfx', t')\right] \times \nonumber \\ [2ex]
&&\qquad\qquad\times \left[H'_{lm} \;+\; S(t-t')\delta_{l2}\,H'_{1m}\right]\nonumber\\[2ex]
&& \label{emf}
\end{eqnarray}
\noindent
Here, $(\widehat{\alpha}\,,\widehat{\beta}\,,\widehat{\eta}\,)$, are {\it transport coefficients}, which
are defined in terms of the $\bfu\bfu$ velocity correlators by
\begin{eqnarray}
\widehat{\alpha}_{il}(\bfx, t, \bfx', t') &\;=\;& \epsilon_{ijm}\left<u_j(\bfx, t) \,u_{ml}(\bfx', t')\right>\nonumber\\[1ex]
\widehat{\beta}_{il}(\bfx, t, \bfx', t') &\;=\;& \epsilon_{ij2}\left<u_j(\bfx, t) \,u_{1l}(\bfx', t')\right>\nonumber\\[1ex]
\widehat{\eta}_{iml}(\bfx, t, \bfx', t') &\;=\;& \epsilon_{ijl}\left<u_j(\bfx, t)  \,u_m(\bfx', t')\right>
\label{trcoeffs}
\end{eqnarray}
\noindent
It is also useful  to consider velocity statistics in terms of $\bfv\bfv$ velocity correlators, because this 
is referred to the lab frame. By definition, from (equation~\ref{newvar}),
\beq
u_m(\bfx, t) \;=\; v_m(\bfX(\bfx,t), t)
\label{uvtr}
\eeq
\noindent
where
\beq
X_1 \;=\; x_1\,,\qquad X_2 \;=\; x_2 + Stx_1\,,\qquad X_3 \;=\; x_3\,,\qquad \tau \;=\; t
\label{invshtr}
\eeq
\noindent
is the inverse of the shearing transformation given in equation~(\ref{sheartr}). Using
\beq
\frac{\partial}{\partial x_l} \;=\; \frac{\partial}{\partial X_l} \;+\; S\tau\,\delta_{l1}\,\frac{\partial}{\partial X_2}
\label{invpartial}
\eeq
\noindent
the velocity gradient $u_{ml}$ can be written as
\beq
u_{ml} \;\equiv\; \frac{\partial u_m}{\partial x_l}
\;=\; \left(\frac{\partial}{\partial X_l} \;+\; S\tau\,\delta_{l1}\,\frac{\partial}{\partial X_2}\right)\,v_m \;=\; v_{ml} \;+\; S\tau\,\delta_{l1}\,v_{m2}
\label{uvgrad}
\eeq
\noindent
where $v_{ml} = (\partial v_m/\partial X_l)$. Then the transport coefficients are given in terms of the 
$\bfv\bfv$ velocity correlators by 
\begin{eqnarray}
\widehat{\alpha}_{il}(\bfx, t, \bfx', t') &\;=\;& \epsilon_{ijm}\left[\left<v_j(\bfX, t)\,v_{ml}(\bfX', t')\right>
\;+\; St'\,\delta_{l1}\,\left<v_j(\bfX, t)\,v_{m2}(\bfX', t')\right>\right]\nonumber\\[1ex]
\widehat{\beta}_{il}(\bfx, t, \bfx', t') &\;=\;& \epsilon_{ij2}\left[\left<v_j(\bfX, t)\,v_{1l}(\bfX', t')\right>
\;+\; St'\,\delta_{l1}\,\left<v_j(\bfX, t)\,v_{12}(\bfX', t')\right>\right]\nonumber\\[1ex]
\widehat{\eta}_{iml}(\bfx, t, \bfx', t') &\;=\;& \epsilon_{ijl}\left<v_j(\bfX, t)\,v_m(\bfX', t')\right>
\label{trcoeffsvv}
\end{eqnarray}
\noindent
where $\bfX$ and $\bfX'$ are shorthand for
\beq
\bfX \;=\; \left(x_1\,,x_2 + Stx_1\,,x_3\right)\,,\qquad\bfX' \;=\; \left(x_1'\,,x_2' + St'x_1'\,,x_3'\right)
\label{xxprime}
\eeq
Equation~(\ref{emf}), together with (\ref{trcoeffs}) or (\ref{trcoeffsvv}), gives the mean EMF in 
general form. $\bfX$ can be thought of as the coordinates of the origin at time $t$ of an observer {\it comoving} with the background shear flow, who was at $\bfx$ at time equal to zero. Similarly, $\bfX'$
can be thought of as the coordinates of the origin at time $t'$ of an observer {\it comoving} with the background shear flow, who was at $\bfx'$ at time equal to zero. 

\section{Galilean--invariant velocity statistics}

\subsection{Galilean invariance of the induction equation}

The induction equation~(\ref{indeqn}) for the total magnetic field --- and also equations~(\ref{meanindeqn}) and (\ref{flucindeqn}) for the mean and fluctuating components --- have a fundamental invariance property
relating to measurements made by a special subset of all observers, called {\it comoving observers} in \cite{SS09a,SS09b}. A comoving observer translates with the velocity of the background shear flow, and 
such an observer can be labelled by the coordinates, $\bfxi = (\xi_1, \xi_2, \xi_3)$, 
of her origin at time $\tau=0$. At any time $\tau$, the origin is at position, 
\beq
\bfX_c(\bfxi,\tau) \;=\; \left(\xi_1\,, \xi_2 + S\tau \xi_1\,, \xi_3\right)
\label{orgvector}
\eeq
\noindent
An event with spacetime coordinates $(\bfX, \tau)$ in the lab frame has spacetime coordinates 
$(\tilde{\bfX}, \tilde{\tau})$ with respect to the comoving observer, given by
\beq
\tilde{\bfX} \;=\; \bfX \;-\; \bfX_c(\bfxi,\tau)\,,\qquad \tilde{\tau} \;=\; \tau - \tau_0
\label{coordtr}
\eeq
\noindent 
where the arbitrary constant $\tau_0$ allows for translation in time as well. 

Let $\left[\tilde{\bfB^{\rm tot}}(\tilde{\bfX}, \tilde{\tau})\,, \tilde{\bfB}(\tilde{\bfX}, \tilde{\tau})\,, \tilde{\bfb}(\tilde{\bfX}, \tilde{\tau})\,,\tilde{\bfv}(\tilde{\bfX}, \tilde{\tau})\right]$ denote the 
total, the mean, the fluctuating magnetic fields and the fluctuating velocity field, respectively, as 
measured by the comoving observer. As explained in \cite{SS09a,SS09b}, these are all equal to the respective quantities measured in the lab frame:
\beq
\left[\tilde{\bfB^{\rm tot}}(\tilde{\bfX}, \tilde{\tau})\,, \tilde{\bfB}(\tilde{\bfX}, \tilde{\tau})\,,
\tilde{\bfb}(\tilde{\bfX}, \tilde{\tau})\,,\tilde{\bfv}(\tilde{\bfX}, \tilde{\tau})\right] \;=\; 
\left[\bfB^{\rm tot}(\bfX, \tau)\,,\bfB(\bfX, \tau)\,,\bfb(\bfX, \tau)\,,\bfv(\bfX, \tau)\right]
\label{fields}
\eeq
\noindent
It is proved in \cite{SS09a,SS09b} that equations~(\ref{indeqn}), (\ref{meanindeqn}) and (\ref{flucindeqn}) are invariant under the simultaneous transformations given in equations~(\ref{coordtr}) and (\ref{fields}). This symmetry property is actually invariance under a subset of the full ten--parameter Galilean group, parametrized
by the five quantities $\left(\xi_1, \xi_2, \xi_3, \tau_0, S\right)$; for brevity we will refer to this restricted symmetry as Galilean invariance, or simply GI. 

It is important to note that the lab and comoving frames need not constitute inertial coordinate
systems. One of the main applications of our theory is to the {\it shearing sheet}, which is a 
local description of a differentially rotating disk. In this case the velocity field will be affected 
by Coriolis forces. The only requirement is that the magnetic field satisfies the induction equation~(\ref{indeqn}). 

\subsection{Galilean--invariant velocity correlators}

We now explore the consequences of requiring that the statistics of 
the velocity fluctuations be Galilean--invariant. We consider the $n$--point velocity correlator measured by the observer in the lab frame. Let this observer
correlate $v_{j_1}$ at spacetime location $(\bfR_1, \tau_1)$, with 
$v_{j_2}$ at spacetime location $(\bfR_2, \tau_2)$, and so on upto
$v_{j_n}$ at spacetime location $(\bfR_n, \tau_n)$. Now consider a comoving observer,
the position vector of whose origin is given by $\bfX_c(\bfxi,\tau)$ 
of equation~(\ref{orgvector}). An identical experiment performed by 
this observer must yield the same results, the measurements now made
at the spacetime points denoted by $\left(\bfR_1 +\bfX_c(\bfxi,\tau_1), \tau_1\right);
\left(\bfR_2 +\bfX_c(\bfxi,\tau_2), \tau_2\right);\ldots\,;\left(\bfR_n +\bfX_c(\bfxi,\tau_n), \tau_n\right)$.
If the velocity statistics is GI, the $n$--point velocity correlator must satisfy the condition
\beq
\left<v_{j_1}(\bfR_1, \tau_1)\,\ldots v_{j_n}(\bfR_n, \tau_n)\right> \;=\; 
\left<v_{j_1}(\bfR_1 + \bfX_c(\bfxi,\tau_1), \tau_1)\,\ldots v_{j_n}(\bfR_n + \bfX_c(\bfxi,\tau_n), \tau_n)\right> 
\label{ginvacorrn}
\eeq
\noindent
for all $(\bfR_1,\ldots\bfR_n\,; \tau_1,\ldots\tau_n\,; \bfxi)$. 

In the low $\rem$ limit, we require only the two--point velocity correlators, for which
\beq
\left<v_i(\bfR, \tau)\,v_j(\bfR', \tau')\right> \;=\; 
\left<v_i(\bfR + \bfX_c(\bfxi,\tau), \tau)\,v_j(\bfR' + \bfX_c(\bfxi,\tau'), \tau')\right> 
\label{ginvacorr}
\eeq
\noindent
for all $(\bfR, \bfR', \tau, \tau', \bfxi)$. We also need to work out the correlation between
velocities and their gradients:
\begin{eqnarray}
\left<v_i(\bfR, \tau)\,v_{jl}(\bfR', \tau')\right> &\;=\;&  
\frac{\partial}{\partial R'_l}\left<v_i(\bfR, \tau)\,v_j(\bfR', \tau')\right>\nonumber\\[2ex]
&\;=\;& \frac{\partial}{\partial R'_l}\left<v_i(\bfR + \bfX_c(\bfxi,\tau), \tau)\,v_j(\bfR' + \bfX_c(\bfxi,\tau'), \tau')\right>\nonumber\\[2ex]
&\;=\;& \left<v_i(\bfR + \bfX_c(\bfxi,\tau), \tau)\,v_{jl}(\bfR' + \bfX_c(\bfxi,\tau'), \tau')\right>
\label{ginvder}
\end{eqnarray}

\begin{figure}
\centerline{\includegraphics[scale=0.5]{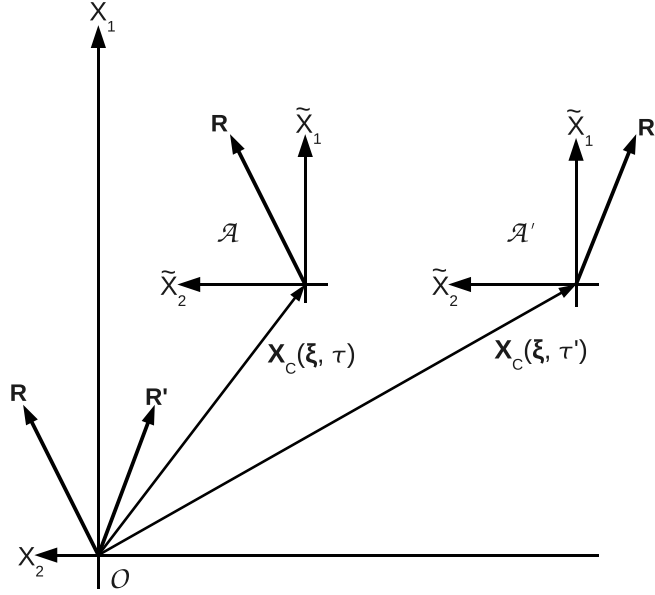}}
\caption{Galilean invariance of the two--point velocity correlator given in equation~(\ref{ginvacorr}). {\cal O} labels 
the observer in the laboratory frame who correlates the velocity fluctuation at location $\bfR$ at time $\tau$ with 
the velocity fluctuation at location $\bfR'$ at a later time $\tau'$. ${\cal A}$ and ${\cal A}'$ label  a comoving 
observer the origin of whose cooordinate axes is at $\bfxi$ at the initial time, and who makes an equivalent measurement 
at the times $\tau$ and $\tau'$.}
\label{gal_inv}
\end{figure}

\noindent
We want to choose $(\bfR, \bfR', \tau, \tau', \bfxi)$ as functions of $(\bfx, \bfx', t, t')$ such that 
we can use equations~(\ref{ginvacorr}) and (\ref{ginvder}) to simplify the velocity correlators in equation~(\ref{trcoeffsvv}). We note that equations~(\ref{xxprime}) and (\ref{orgvector}) give
\beq
\bfX \,=\, \bfX_c(\bfx,t) \, , \qquad\bfX' \,=\, \bfX_c(\bfx',t')
\eeq
\noindent
It is therefore natural to choose
\beq
\tau \,=\, t \, , \qquad \tau' \,=\, t'
\eeq
\noindent
Thus the velocity correlators we require can now be written as
\begin{eqnarray}
 \left<v_i(\bfX, t)\,v_j(\bfX', t')\right> &\;=\;& 
\left<v_i(\bfX_c(\bfx,t), t)\,v_j(\bfX_c(\bfx',t'), t')\right> \nonumber \\ [3ex]
\left<v_i(\bfX, t)\,v_{jl}(\bfX', t')\right> &\;=\;& 
\left<v_i(\bfX_c(\bfx,t), t)\,v_{jl}(\bfX_c(\bfx',t'), t')\right>
\label{ginvv1}
\end{eqnarray}
\noindent
Comparing equation~(\ref{ginvv1}) with equations~(\ref{ginvacorr}) and (\ref{ginvder}), we see that if 
we choose
\beq
\bfR = \bfX_c(\bfx,t), \qquad\bfR' = \bfX_c(\bfx',t')
\eeq
\noindent
then equation~(\ref{ginvv1}), together with equations (\ref{orgvector}), (\ref{ginvacorr}) and (\ref{ginvder}) implies that
\begin{eqnarray}
\left< v_i (\bfX,t) v_j (\bfX',t')\right> &=& \left< v_i(\bfR,\tau)v_j(\bfR',\tau')\right> \nonumber \\ [2ex]
&=& \left< v_i(\bfR+\bfX_c(\bfxi,\tau),\tau)v_j(\bfR'+\bfX_c(\bfxi,\tau'),\tau') \right> \nonumber \\ [2ex]
&=& \left< v_i(\bfX_c(\bfx+\bfxi,t),t)v_j(\bfX_c(\bfx'+\bfxi,t'),t') \right> \nonumber 
\end{eqnarray}
\noindent
Similarly
\begin{eqnarray}
\left< v_i(\bfX,t)v_{jl}(\bfX',t') \right> &=& \left< v_i(\bfX_c(\bfx+\bfxi,t),t)v_{jl}(\bfX_c(\bfx'+\bfxi,t'),t') \right>
\end{eqnarray}
\noindent
Now it is natural to choose
\beq
\bfxi = -\frac{1}{2}(\bfx+\bfx')
\eeq
\noindent
Then
\begin{eqnarray}
\left< v_i(\bfX,\tau)v_j(\bfX',\tau')\right> &=& \left< v_i \left(\bfX_c \left(\frac{\bfx-\bfx'}{2},t \right),t \right) v_j \left(\bfX_c \left(\frac{\bfx'-\bfx}{2},t' \right),t' \right) \right>\nonumber\\[3ex]
&=& R_{ij}(\bfx-\bfx',t,t')
\label{rdef}
\end{eqnarray}
\noindent
Similarly,
\begin{eqnarray}
\left< v_i(\bfX,\tau)v_{jl}(\bfX',\tau')\right> &=& \left< v_i \left(\bfX_c \left(\frac{\bfx-\bfx'}{2},t \right),t \right) v_{jl} \left(\bfX_c \left(\frac{\bfx'-\bfx}{2},t' \right),t' \right) \right>\nonumber\\[3ex]
&=& Q_{ijl}(\bfx-\bfx',t,t')
\label{sdef}
\end{eqnarray}
\noindent
We note that symmetry and incompressibility imply that
\begin{eqnarray}
R_{ij}(\bfr,t,t') \;&=&\; R_{ji}(-\bfr,t',t)\nonumber\\[2ex]
Q_{ijj}(\bfr,t,t') \;&=&\; 0
\end{eqnarray}

\subsection{Galilean--invariant mean EMF}

The transport coefficients are completely determined by the form of the velocity correlator.
Using equations~(\ref{rdef}) and (\ref{sdef}) in equations~(\ref{trcoeffsvv}) and noting the fact that the velocity correlators defined above are functions only of $(\bfx - \bfx')$, $t$ and $t'$, we can see that the GI transport coefficients, 
\begin{eqnarray}
\widehat{\alpha}_{il}(\bfx,t,\bfx', t') &\;=\;& \epsilon_{ijm}\left[Q_{jml}(\bfx-\bfx',t,t') \;+\; St'\,\delta_{l1}\,Q_{jm2}(\bfx-\bfx',t,t')\right]\nonumber\\[1ex]
\widehat{\beta}_{il}(\bfx,t,\bfx',t') &\;=\;& \epsilon_{ij2}\left[Q_{j1l}(\bfx-\bfx',t,t') \;+\; St'\,\delta_{l1}\,Q_{j12}(\bfx-\bfx',t,t')\right]\nonumber\\[1ex]
\widehat{\eta}_{iml}(\bfx,t,\bfx',t') &\;=\;& \epsilon_{ijl}\,R_{jm}(\bfx-\bfx',t, t')
\label{trcoeffginv}
\end{eqnarray}
\noindent
The transport coefficients depend on $\bfx$ and $\bfx'$ only through the combination, $(\bfx - \bfx')$, 
which arises because of Galilean invariance. We can derive an expression for the G--invariant mean EMF by using equations~(\ref{trcoeffginv}) for the transport coefficients in equation~(\ref{emf}). We also change the integration variable in equation~(\ref{emf}) to $\bfr \,=\, \bfx-\bfx'$. The integrands can be simplified as follows:
\begin{eqnarray}
\widehat{\alpha}_{il}(\bfx,t,\bfx', t')\left[H'_l - St'\delta_{l2}H'_1\right]
&\;=\;& \epsilon_{ijm}\left[Q_{jml} \;+\; St'\,\delta_{l1}\,Q_{jm2}\right]
 \left[H'_l - St'\delta_{l2}H'_1 \right]\nonumber\\[1ex]
&\;=\;& \epsilon_{ijm}Q_{jml}(\bfr,t,t')H_l(\bfx-\bfr,t')\nonumber\\[2ex]
\widehat{\beta}_{il}(\bfx,t,\bfx', t')\left[H'_l - St'\delta_{l2}H'_1\right]
&\;=\;& \epsilon_{ij2}\left[Q_{j1l} \;+\; St'\,\delta_{l1}\,Q_{j12}\right] 
\left[H'_l - St'\delta_{l2}H'_1\right]\nonumber\\[1ex]
&\;=\;& \epsilon_{ij2}Q_{j1l}(\bfr,t,t')H_l(\bfx-\bfr,t')\nonumber
\end{eqnarray}

\begin{eqnarray}
\left[\,\widehat{\eta}_{iml} - St'\delta_{m2}\,\widehat{\eta}_{i1l}\right]H'_{lm}
&\;=\;& \epsilon_{ijl}\left[R_{jm}(\bfr,t, t') \;-\; St'\delta_{m2}\,R_{j1}(\bfr,t,t') \right]H_{lm}(\bfx-\bfr,t') \nonumber\\[2ex]
\left[\,\widehat{\eta}_{im2} - St'\delta_{m2}\,\widehat{\eta}_{i12}\right]H'_{1m}
&\;=\;& \epsilon_{ij2}\,\delta_{l1}\left[R_{jm}(\bfr,t, t') \;-\; St'\delta_{m2}\,R_{j1}(\bfr,t, t') \right]H_{lm}(\bfx-\bfr,t') \nonumber
\end{eqnarray}

\noindent
Define
\begin{eqnarray}
C_{jml}(\bfr,t,t') &\;=\;& Q_{jml}(\bfr,t,t') \;+\; S(t-t')\delta_{m2}\,Q_{j1l}(\bfr,t,t')\nonumber\\[2ex]
D_{jm}(\bfr,t,t') &\;=\;& R_{jm}(\bfr,t,t') \;-\; St'\delta_{m2}\,R_{j1}(\bfr,t,t')
\label{cddef}
\end{eqnarray}
\noindent
The mean EMF can now be written compactly as
\begin{eqnarray}
{\cal E}_i(\bfx, t)
\;=\; && \epsilon_{ijm}\,\int_0^t \mathrm{d}t' \int \mathrm{d}^3r \; G_{\eta}(\bfr,t,t')\,C_{jml}(\bfr,t,t')H_l(\bfx-\bfr,t') \nonumber \\ [2ex]
\;&-&\; \int_0^t \mathrm{d}t' \int \mathrm{d}^3r \, G_{\eta}(\bfr,t,t')\left[\epsilon_{ijl} + S(t-t')\delta_{l1}\epsilon_{ij2}\right] \times \nonumber \\[2ex] 
&& \qquad \qquad \qquad \qquad\times D_{jm}(\bfr,t, t')H_{lm}(\bfx-\bfr,t')
\label{emfcd}
\end{eqnarray}

\section{Mean--field induction equation}

\subsection{Mean--field induction equation in sheared coordinate space}

Applying the shearing transformation given in equations~(\ref{sheartr}) and (\ref{partialsh}) to
the mean--field equation~(\ref{meanindeqn}), we see that the mean--field, $\bfH(\bfx, t)$, obeys
\beq
\frac{\partial H_i}{\partial t} \;-\; S\delta_{i2}H_1 \;=\;
\left(\bnabla\cross\bfemf\right)_i \;+\; \eta\bnabla^2 H_i
\label{hmeanindeqn}
\eeq
\noindent
where 
\beq
\left(\bnabla\right)_p \;\equiv\; \frac{\partial}{\partial X_p} \;=\; 
\frac{\partial}{\partial x_p} \;-\; St\,\delta_{p1}\frac{\partial}{\partial x_2}
\label{deltran}
\eeq
\noindent
We note that the divergence condition on the mean magnetic field can be written as
\beq
\bnabla\cendot\bfH \;\equiv\; \frac{\partial H_p}{\partial X_p} \;=\; H_{pp} \;-\; StH_{12} \;=\; 0
\label{divcond}
\eeq
\noindent
It may be verified that equation~(\ref{hmeanindeqn}) preserves the condition $\bnabla\cendot\bfH = 0\,$.
We now use equations~(\ref{emfcd}) and (\ref{deltran}) to evaluate $\bnabla\cross\bfemf$.
\begin{eqnarray}
\left(\bnabla\cross\bfemf\right)_i &=& \epsilon_{ipq}\frac{\partial{\cal E}_q}{\partial X_p}
\;=\; \epsilon_{ipq}\left(\frac{\partial}{\partial x_p} \;-\; St\,\delta_{p1}\frac{\partial}{\partial x_2}
\right){\cal E}_q
\nonumber\\[3ex] 
&=& \epsilon_{ipq}\epsilon_{qjm}\,\int_0^t \mathrm{d}t' \int \mathrm{d}^3r \; G_{\eta}(\bfr,t,t')C_{jml}(\bfr,t,t')
\left[H'_{lp} - St\,\delta_{p1}H'_{l2}\right]\nonumber\\[2ex]
&&- \int_0^t \mathrm{d}t' \int \mathrm{d}^3r \, G_{\eta}(\bfr,t,t') D_{jm}(\bfr,t, t')\left[\epsilon_{ipq}\epsilon_{qjl} + S(t-t')\delta_{l1}\epsilon_{ipq}\epsilon_{qj2}
\right] \times \nonumber \\ [1ex]
&& \qquad \qquad \qquad \times \left[H'_{lmp} - St\,\delta_{p1}H'_{lm2}\right]
\nonumber
\end{eqnarray}
\noindent
where $H'_i = H_i(\bfx - \bfr, t')$.  Expanding $\epsilon_{ipq}\epsilon_{qjm} = \left(\delta_{ij}\,\delta_{mp} - \delta_{im}\,\delta_{jp}\right)$, 
the contribution from the $C$ term is
\begin{eqnarray}
\left(\bnabla\cross\bfemf\right)^C_i \;=\; && \, \int_0^t \mathrm{d}t' \int \mathrm{d}^3r\; G_{\eta}(\bfr,t,t') \left[C_{ipl} (\bfr,t,t') - C_{pil} (\bfr,t,t')\right] \times \nonumber \\ [2ex]
&& \qquad \qquad \times \left[H'_{lp} - St\delta_{p1}H'_{l2}\right]
\label{curlemfc}
\end{eqnarray}
\noindent
Evaluating the $D$ term is a bit more involved. Again, we begin by expanding $\epsilon_{ipq}\epsilon_{qjl} = \left(\delta_{ij}\,\delta_{lp} - \delta_{il}\,\delta_{jp}\right)$. Then we get 
\begin{eqnarray}
\left(\bnabla\cross\bfemf\right)^D_i &=&
\int_0^t \mathrm{d}t'  \int \mathrm{d}^3r\; G_{\eta}(\bfr,t,t') \,D_{pm} (\bfr, t, t') \times \nonumber \\ [1ex]
&& \qquad \times \left\{H'_{ipm} - St\delta_{p1}H'_{i2m} + 
S(t-t')\delta_{i2}\left[H'_{1pm} - St\delta_{p1}H'_{12m}\right]\right\}\nonumber\\[2ex]
&&- \int_0^t \mathrm{d}t'  \int \mathrm{d}^3r\; G_{\eta}(\bfr,t,t') \,D_{im} (\bfr,t,t')\left[H'_{ppm} - St'H'_{12m}\right]
\label{curlemfd}
\end{eqnarray}
\noindent
The second integral vanishes because the factor in $[\;]$ multiplying $D_{im}$ is zero: to see this,
differentiate the divergence--free condition of equation~(\ref{divcond}) with respect to $x_m$. 
We can now use equations~(\ref{curlemfc}) and (\ref{curlemfd}) to write $\left(\bnabla\cross\bfemf\right) = \left(\bnabla\cross\bfemf\right)^C + \left(\bnabla\cross\bfemf\right)^D$. Substituting this expression in 
equation~(\ref{hmeanindeqn}), we obtain a set of integro--differential equation governing the dynamics of the mean--field, $\bfH(\bfx, t)$, valid for arbitrary values of the shear strength $S$:
\begin{eqnarray}
\frac{\partial H_i}{\partial t} \;-\; S\delta_{i2}H_1 \;=\; 
\eta\bnabla^2 H_i &+& 
\int_0^t \mathrm{d}t' \int \mathrm{d}^3r\; G_{\eta}(\bfr,t,t') \; \left[C_{iml} (\bfr,t,t')- C_{mil} (\bfr,t,t')\right]\times\nonumber\\[2ex]
&&\times \left[H_{lm} (\bfx - \bfr, t') \;-\; St\delta_{m1}H_{l2} (\bfx - \bfr, t')\right]  \nonumber\\[4ex]
&+& \int_0^t \mathrm{d}t'  \int \mathrm{d}^3r\; G_{\eta}(\bfr,t,t') \,D_{jm} (\bfr, t, t')\times\nonumber\\[2ex] 
&&\times\left[H_{ijm}(\bfx - \bfr, t') \;-\; St\delta_{j1}H_{i2m}(\bfx - \bfr, t') \;+\;\right.\nonumber\\[2ex]
&&\left. +S(t-t')\delta_{i2}\left\{H_{1jm}(\bfx - \bfr, t') - St\delta_{j1}H_{12m}(\bfx - \bfr, t')\right\}\right]\nonumber\\ 
&& \label{hmeaneqnfin}
\end{eqnarray}
\noindent
We note some important properties of the mean--field induction equation~(\ref{hmeaneqnfin}):
\begin{enumerate}

\item The $D_{jm}(\bfr,t,t')$ terms are such that $\left(\bnabla\cross\bfemf\right)_i$
involves only $H_i$ for $i=1$ and $i=3$, whereas $\left(\bnabla\cross\bfemf\right)_2$
depends on both  $H_2$ and $H_1$.  The implications for the original field, $\bfB(\bfX, \tau)$, 
can be read off, because it is equal to $\bfH(\bfx, t)$ component--wise (i.e $B_i(\bfX, \tau) = H_i(\bfx, t)$). 
Therefore, in the mean--field induction equation, the ``D'' terms are of such a form that: (i) the equations for 
$B_1$ or $B_3$ involve only $B_1$ or $B_3$, respectively; (ii) the equation for $B_2$ involves both 
$B_1$ and $B_2$. 

\item Only the part of $C_{iml}(\bfr,t,t')$ that is antisymmetric in the indices $(i,m)$ 
contributes. We note that it is possible that the  ``C'' terms can lead to a coupling of different components 
of the mean magnetic field. To investigate this, it is necessary to specify the statistics of the velocity fluctuations.

\item In the formal limit of zero resistivity, $\eta\to 0$, the resistive Green's function, 
$G(\bfx, t, t')\to\delta(\bfx)$. Then the mean--field induction equation simplifies to
\begin{eqnarray}
\frac{\partial H_i}{\partial t} \;-\; S\delta_{i2}H_1 \;=\; 
\quad &&\int_0^t \mathrm{d}t'\,\left[C_{iml} ({\bf 0},t,t')- C_{mil} ({\bf 0},t,t')\right]\times\nonumber\\[2ex]
&&\qquad\times \left[H_{lm}(\bfx, t') \;-\; St\delta_{m1}H_{l2}(\bfx, t')\right]\nonumber\\[4ex]
&+& \int_0^t \mathrm{d}t'\,D_{jm} ({\bf 0}, t, t')\left[H_{ijm}(\bfx, t') \;-\; St\delta_{j1}H_{i2m}(\bfx, t') \;+\;\right.\nonumber\\[2ex]
&&\left. +S(t-t')\delta_{i2}\left\{H_{1jm}(\bfx, t') - St\delta_{j1}H_{12m}(\bfx,t')\right\}\right]\nonumber\\ 
&& \label{hmeaneqnetazero}
\end{eqnarray}
\noindent
which is identical to that derived in \cite{SS09a,SS09b}. 
\end{enumerate}

\subsection{Mean--field induction equation in sheared Fourier space}

Equation~(\ref{hmeaneqnfin}) governing the time evolution of the mean field may be simplified further
by taking a spatial Fourier transformation. Let us define
\beq
\tilde{\bfH}(\bfk, t) \;=\; \int\,\mathrm{d}^3x\,\bfH(\bfx, t)\exp{(-\mathrm{i}\,\bfk\cendot\bfx)}
\label{fth}
\eeq 
\noindent and the quantities,
\begin{eqnarray}
\tilde{I}_{iml}(\bfk,t,t') &=& \int\,\mathrm{d}^3r\,G_{\eta}(\bfr,t,t')C_{iml}(\bfr,t,t')
\exp{(-\mathrm{i}\,\bfk\cendot\bfr)}\nonumber\\[3ex]
\tilde{J}_{jm}(\bfk,t,t') &=& \int\,\mathrm{d}^3r\,G_{\eta}(\bfr,t,t')D_{jm}(\bfr,t,t')
\exp{(-\mathrm{i}\,\bfk\cendot\bfr)}
\label{ijdef}
\end{eqnarray}
\noindent
Both $\tilde{I}_{iml}(\bfk,t,t')$ and $\tilde{J}_{jm}(\bfk,t,t')$ are to be regarded as given quantities, 
because they are known once the velocity correlators have been specified. Taking the spatial Fourier 
transform of equation~(\ref{hmeaneqnfin}), we obtain, 
\begin{eqnarray}
\frac{\partial \tilde{H}_i}{\partial t} \;-\; S\delta_{i2}\tilde{H}_1 
&=& -\eta K^2\tilde{H}_i \;+\; \mathrm{i}\, K_m \int_0^t \mathrm{d}t'\,\left[\tilde{I}_{iml}(\bfk,t,t')
- \tilde{I}_{mil}(\bfk,t,t')\right]\tilde{H}_l(\bfk, t')\nonumber\\[3ex]
&& -k_mK_j\int_0^t \mathrm{d}t'\,\tilde{J}_{jm}(\bfk,t,t')\left[\tilde{H}_i(\bfk, t')
+S(t-t')\delta_{i2}\tilde{H}_1(\bfk, t')\right]\nonumber\\
&&\label{ftindeqn}
\end{eqnarray}
\noindent
where $\bfK(\bfk, t) = (k_1 - St\,k_2, \,k_2, \,k_3)$ and $K^2 = \vert\bfK\vert^2 =
(k_1 - Stk_2)^2 + k_2^2 + k_3^2\,$. Once the initial data, $\tilde{\bfH}(\bfk, 0)$, has been specified, equations~(\ref{ftindeqn}) can be integrated in time to determine $\tilde{\bfH}(\bfk, t)$. Whereas these equations are not easy to solve, we note some of their important properties: 

\begin{enumerate}

\item Only the part of $\tilde{I}_{iml}(\bfk,t,t')$ that is antisymmetric in the indices $(i,m)$ 
contributes. 

\item The time evolution of $\tilde{\bfH}(\bfk, t)$ depends only on $\tilde{\bfH}(\bfk, t')$ for 
$0\leq t'< t$, not on the values of $\tilde{\bfH}$ at other values of $\bfk$. Thus each $\bfk$ labels 
a {\it normal mode} whose {\it amplitude} and {\it polarisation} are given by $\tilde{\bfH}(\bfk, t)$, the  time evolution of which is independent of all the other normal modes. 

\item When we have determined  $\tilde{\bfH}(\bfk, t)$, the magnetic field in the original variables, $\bfB(\bfX, \tau)$, can be recovered by 
using the shearing transformation, equation~(\ref{sheartr}), to write $(\bfx, t)$ in terms of the lab frame coordinates $(\bfX, \tau)$:
\begin{eqnarray}
\bfB(\bfX, \tau) &=& \bfH(\bfx, t) \;=\;
\int\,\frac{\mathrm{d}^3k}{(2\pi)^3}\,\tilde{\bfH}(\bfk, t)\exp{(\mathrm{i}\,\bfk\cendot\bfx)}\nonumber\\[3ex]
&=& \int\,\frac{\mathrm{d}^3k}{(2\pi)^3}\,\tilde{\bfH}(\bfk, \tau)\exp{\left[\mathrm{i}\,\bfK(\bfk, \tau)\cendot\bfX\right]}
\label{bmeanexp}
\end{eqnarray} 
\noindent
where we have used $\bfK\cendot\bfX =\bfk\cendot\bfx$. Thus $\bfB(\bfX, \tau)$ has been expressed as a superposition of the normal modes, each of which is a
{\it shearing wave}, whose spatial structure is given by 
\beq
\exp{\left[\mathrm{i}\,\bfK(\bfk, \tau)\cendot\bfX\right]} \;=\; \exp{\left[\mathrm{i}\,\{(k_1 - Stk_2)X_1 + k_2X_2 + k_3X_3\}\right]}
\label{shwave}
\eeq
\noindent
For non--axisymmetric waves, $k_2\neq 0$ and, as time progresses, the shearing wave develops fine-structure along the $X_1$--direction with a time--dependent spatial frequency equal to $(k_1 - Stk_2)$. 
\end{enumerate}

\subsection{The integral kernels expressed in terms of the velocity spectrum tensor}

We have derived the integral equation satisfied by the mean magnetic field, to lowest order in $\rem$; in sheared coordinate space it is 
given by equations~(\ref{hmeaneqnfin}), and in sheared Fourier space it is given by equations~(\ref{ftindeqn}). One can proceed to look 
for solutions if the integral kernels are known. This means that either the pair $\left[C_{iml} (\bfr,t,t')\,, D_{jm}(\bfr,t,t')\right]$ or the pair 
$\left[\tilde{I}_{iml}(\bfk,t,t')\,, \tilde{J}_{jm}(\bfk,t,t')\right]$ needs to be specified. Here we show that all these integral kernels can be expressed 
in terms of a single entity, which is the velocity spectrum tensor, $\Pi_{ij}(\bfk, t, t')$.
The Galilean invariance of velocity correlators stated in equation~(\ref{ginvacorr}) is most compactly expressed in Fourier--space; this is stated in 
the theorem below. Let $\tilde{\bfv}(\bfK, \tau)$ be the spatial Fourier transform of $\bfv(\bfX,\tau)$, defined by
\beq
\tilde{\bfv}(\bfK, \tau) \;=\; \int\, \mathrm{d}^3X\exp{(-\mathrm{i}\,\bfK\cdot\bfX)}\, \bfv(\bfX, \tau)\,;\qquad\qquad \left[\bfK\cdot\tilde{\bfv}(\bfK, \tau)\right] \;=\; 0
\eeq 
\noindent
It is proved in the Appendix that a G--invariant Fourier--space two--point velocity correlator must be of the form
\beq
\left<\tilde{v}_i(\bfK, \tau)\,\tilde{v}_j^*(\bfK', \tau')\right> \;=\; 
(2\pi)^6\,\delta(\bfk -\bfk')\,\Pi_{ij}(\bfk, t, t')
\label{spec}
\eeq

We first work out $R_{jm}(\bfr, t, t')$ and $Q_{jml}(\bfr, t, t')$ in terms of $\Pi_{jm}(\bfk, t, t')$. From equation~(\ref{rdef}) and  (\ref{sdef}) we have
\begin{eqnarray}
R_{jm}(\bfr,t,t') &=& \left< v_j\left(\bfX_c \left(\frac{\bfr}{2},t \right),t \right) v_m \left(\bfX_c \left(-\frac{\bfr}{2},t' \right),t' \right) \right>\nonumber\\[3ex]
&=& \int\,\frac{\mathrm{d}^3K}{(2\pi)^3}\,\frac{\mathrm{d}^3K'}{(2\pi)^3}\,\left<v_j\left(\bfK,t\right)v_m^*\left(\bfK',t'\right)\right>\,\times\nonumber\\[1ex]
&&\qquad\times\,\exp{\left[\mathrm{i}\,\left(\bfK\cendot\bfX_c\left(\frac{\bfr}{2},t\right) - \bfK'\cendot\bfX_c\left(-\frac{\bfr}{2},t'\right)\right)\right]}
\label{rfour}
\end{eqnarray}
\begin{eqnarray}
Q_{jml}(\bfr,t,t') &=& \left< v_j\left(\bfX_c \left(\frac{\bfr}{2},t \right),t \right) v_{ml} \left(\bfX_c \left(-\frac{\bfr}{2},t' \right),t' \right) \right>\nonumber\\[3ex]
&=& \int\,\frac{\mathrm{d}^3K}{(2\pi)^3}\,\frac{\mathrm{d}^3K'}{(2\pi)^3}\,\left(-\mathrm{i}\,K'_l\right)\,
\left<v_j\left(\bfK,t\right)v_m^*\left(\bfK',t'\right)\right>\,\times\nonumber\\[1ex]
&&\qquad\times\,\exp{\left[\mathrm{i}\,\left(\bfK\cendot\bfX_c\left(\frac{\bfr}{2},t\right) - \bfK'\cendot\bfX_c\left(-\frac{\bfr}{2},t'\right)\right)\right]}\nonumber\\
&&\label{sfour}
\end{eqnarray}
\noindent
Substituting for $\bfX_c$ from equation~(\ref{orgvector}), we can write the phase
\beq 
\bfK\cendot\bfX_c\left(\frac{\bfr}{2},t\right) - \bfK'\cendot\bfX_c\left(-\frac{\bfr}{2},t'\right) \;=\; \left(\bfk + \bfk'\right)\cendot\frac{\bfr}{2}
\label{phasesimple}
\eeq
\noindent
where $\bfk$ and $\bfk'$ are defined in  equations~(\ref{fourshtr}). Therefore,
\begin{eqnarray}
R_{jm}(\bfr,t,t') &=& \int\,\mathrm{d}^3k\,\Pi_{jm}(\bfk,t,t')\,\exp{\left[\mathrm{i}\,\bfk\cendot\bfr\right]}\nonumber\\[5ex]
Q_{jml}(\bfr,t,t') &=& -\mathrm{i}\,\int\,\mathrm{d}^3k\,\left[k_l - St'\delta_{l1}k_2\right]\,\Pi_{jm}(\bfk,t,t')\,\exp{\left[\mathrm{i}\,\bfk\cendot\bfr\right]}
\label{rsfour}
\end{eqnarray}
\noindent
Using equations~(\ref{cddef}) we can write the real--space integral kernels, $C_{jml}(\bfr,t,t')$ and $D_{jm}(\bfr,t,t')$, as
\begin{eqnarray}
D_{jm}(\bfr,t,t') &=& R_{jm}(\bfr,t,t') \;-\; St'\delta_{m2}\,R_{j1}(\bfr,t,t')\nonumber\\[3ex]
&=& \int\,\mathrm{d}^3k\,\left[\Pi_{jm}(\bfk,t,t')  - St'\delta_{m2}\,\Pi_{j1}(\bfk,t,t')\right]\exp{\left[\mathrm{i}\,\bfk\cendot\bfr\right]}\nonumber\\[5ex]
C_{jml}(\bfr,t,t') &=& Q_{jml}(\bfr,t,t') \;+\; S(t-t')\delta_{m2}\,Q_{j1l}(\bfr,t,t')\nonumber\\[3ex]
&=& -\mathrm{i}\,\int \mathrm{d}^3k\left[k_l - St'\delta_{l1}k_2\right]\left[\Pi_{jm}(\bfk,t,t') +  S(t-t')\delta_{m2}\Pi_{j1}(\bfk,t,t')\right]
\exp{\left[\mathrm{i}\,\bfk\cendot\bfr\right]}\nonumber\\
&&\label{dcfour}
\end{eqnarray}
\noindent
Using equations~(\ref{ijdef}) we can express the Fourier--space integral kernels, $\tilde{I}_{jml}(\bfk,t,t')$ and $\tilde{J}_{jm}(\bfk,t,t')$, 
as
\begin{eqnarray}
\tilde{J}_{jm}(\bfk,t,t') &=& \int\,\mathrm{d}^3k'\, \tilde{G}_{\eta}(\bfk - \bfk', t, t')\left[\Pi_{jm}(\bfk',t,t')  - St'\delta_{m2}\,\Pi_{j1}(\bfk',t,t')\right]
\nonumber\\[5ex]
\tilde{I}_{jml}(\bfk,t,t') &=& -\mathrm{i}\,\int\, \mathrm{d}^3k'\,\tilde{G}_{\eta}(\bfk - \bfk', t, t')\left[k'_l - St'\delta_{l1}k'_2\right]\,\times\nonumber\\[2ex]
&&\qquad\times\,\left[\Pi_{jm}(\bfk',t,t') +  S(t-t')\delta_{m2}\Pi_{j1}(\bfk',t,t')\right]
\label{jifour}
\end{eqnarray}
Thus, we have expressed the integral kernels in terms of  the velocity spectrum tensor, $\Pi_{jm}(\bfk, t, t')$, which is the fundamental
dynamical quantity that needs to be calculated before the integro--differential equation for the mean magnetic field can be solved.

\section{Conclusions}

We have formulated the problem of large--scale kinematic dynamo action due to turbulence 
in the presence of a linear shear flow, in the limit of small magnetic Reynolds number ($\rem$) 
but arbitrary fluid Reynolds number. The mean--field theory we present is non perturbative in the
shear parameter, and makes systematic use of the shearing coordinate transformation
and the Galilean invariance of the linear shear flow. Using Reynolds averaging, we split the magnetic field into 
mean and fluctuating components. The mean magnetic field is driven by the Curl of the mean EMF, which in 
turn must be determined in terms of the statistics of the velocity fluctuations. In order to do this  
it is necessary to determine the magnetic fluctuations in terms of the mean magnetic field and the velocity fluctuations. 
So we develop the equation for the fluctuations perturbatively in the small parameter, $\rem$. 
Using the shearing coordinate transformation, we make an explicit calculation of the resistive Green's 
function for the linear shear flow. From the perturbative scheme it is clear that the fluctuations can be determined to 
any order in $\rem$. Here we determine the magnetic fluctuations and the mean EMF to lowest order in $\rem$. 
The transport coefficients are given in general form in terms of the two--point correlators of the velocity fluctuations. 
At this point we make use of Galilean invariance, which is a fundamental symmetry of the problem. For Galilean invariant 
velocity statistics we prove that the transport coefficients, although space-dependent, possess the 
property of translational invariance in sheared coordinate space. An explicit expression for the Galilean--invariant mean 
EMF is derived. 

We put together all the results in \textsection~5 by deriving the integro--differential equation governing the time evolution of the mean magnetic field. Some important properties of this equation are the following: 

\begin{enumerate}
\item Velocity fluctuations contribute to two different kinds of terms, the ``C'' and ``D'' terms, in which first and second 
spatial derivatives of the mean magnetic field, respectively, appear inside the spacetime integrals. 

\item The ``C'' terms are a generalization to the case of shear, of the ``$\alpha$'' term familiar from mean--field 
electrodynamics in the absence of shear. However, they can also contribute to ``magnetic diffusion''; see discussion below.
Likewise, the ``D'' terms are a generalization to the case of shear, of the ``magnetic diffusion''  term familiar from mean--field 
electrodynamics in the absence of shear. It must be noted that the generalization is non perturbative in the shear strength. 

\item In the mean--field induction equation, the ``D'' terms are of such a form that: (i) the equations for 
$B_1$ or $B_3$ involve only $B_1$ or $B_3$, respectively; (ii) the equation for $B_2$ involves both 
$B_1$ and $B_2$. Therefore, to lowest order in $\rem$ but to all orders in the shear strength, the ``D'' terms cannot
give rise to  a shear--current assisted dynamo effect.

\item In the formal limit of zero resistivity, the quasilinear theory of \cite{SS09a, SS09b} is recovered. In this case, 
the ``C'' terms vanish when the velocity field is non helical. However, this may not be the case when the resitivity is 
non zero. Whether the ``C'' terms give rise to such a shear--current--type effect depends on the form of the velocity correlators, 
which will be strongly affected by shear and highly anisotropic; hence it is difficult to guess their tensorial forms {\it a priori}
and it is necessary to develop a dynamical theory of velocity correlators -- see below for further discussion.

\item Sheared Fourier space is the natural setting for the mean magnetic field;  the normal modes
of the theory are a set of shearing waves, labelled by their sheared wavevectors. 

\item We prove a theorem (in the Appendix) on the form of the two--point velocity correlator in Fourier space; the velocity spectrum tensor and its 
general properties are discussed. The integral kernels are expressed in terms of the velocity spectrum tensor, which is the fundamental 
dynamical quantity that needs to be specified to complete the integro--differential equation description of the time 
evolution of the mean magnetic field.
\end{enumerate}

The physical meaning of the  ``C'' and ``D'' terms becomes clear in the limit of a slowly varying magnetic field, when the integro--differential equation 
reduces to a partial differential equation \citep{SS10}. Then we  encounter the well--known $\alpha$--effect and turbulent magnetic diffusion ($\eta$), 
albeit in tensorial form. The ``C'' terms alone contributes to $\alpha$, whereas both ``C'' and ``D'' terms contribute to magnetic diffusion. When the 
velocity field is non helical, the velocity spectrum tensor is real, and the tensorial $\alpha$ coefficient vanishes; this result is true for arbitrary values of 
the shear parameter. The ``C'' terms can, in principle, contribute to a shear--current effect, through the off--diagonal components of the diffusivity tensor (which couple the streamwise 
component of the mean magnetic field with the cross--stream components). It turns out that these off--diagonal components depend on the 
microscopic resistivity in such a manner that they vanish when the microscopic resistivity vanishes. This result is consistent with the results of
\cite{SS09a, SS09b}. To deal with the case when the microscopic resistivity does not vanish, it is necessary to provide our kinematic development with 
a dynamical model for the velocity field. \cite{SS10} show that, for forced non helical driving at low fluid Reynolds number, the sign of the off--diagonal 
terms of the diffusivity tensor does not favour the shear--current effect. This conclusion agrees with those reported in \cite{RS06,RK06,BRRK08}, even if  
our results are limited to low Reynolds numbers. If we seek a different explanation for the dynamo action seen in numerical simulations, 
the  ``fluctuating $\alpha$--effect'' still remains a promising candidate. $\alpha$ itself is described by second--order velocity correlators, so to describe 
fluctuatons of $\alpha$, it is necessary to deal with either fourth--order velocity correlators or products of two second--order velocity correlators. 
This requires extending our perturbative calculations by at least two higher orders, a task which, while tractable,  is beyond the scope of 
the present investigation. 

\acknowledgments
We thank Axel Brandenburg, Karl--Heinz R\"adler and Matthias Rheinhardt for valuable comments.

\ifCUPmtlplainloaded \newpage\fi

\appendix
\section{A theorem on Galilean invariant velocity correlators}

\newtheorem{theorem}{Theorem}
\begin{theorem}
A G--invariant Fourier--space two--point velocity correlator must be of the form
\beq
\left<\tilde{v}_i(\bfK, \tau)\,\tilde{v}_j^*(\bfK', \tau')\right> \;=\; 
(2\pi)^6\,\delta(\bfk -\bfk')\,\Pi_{ij}(\bfk, t, t')
\label{speca}
\eeq
\noindent
where $\Pi_{ij}$ is the {\it velocity spectrum tensor}, which must possess the following properties: 
\begin{eqnarray}
\Pi_{ij}(\bfk, t, t') &\;=\;& \Pi_{ij}^*(-\bfk, t, t') \;=\; \Pi_{ji}(-\bfk, t', t)
\nonumber\\[2ex]
K_i\,\Pi_{ij}(\bfk, t, t') &\;=\;& \left(k_i - St\,\delta_{i1}k_2\right)\Pi_{ij}(\bfk, t, t') \;=\; 0
\nonumber\\[2ex]
K'_j\,\Pi_{ij}(\bfk, t, t') &\;=\;& \left(k_j - St'\,\delta_{j1}k_2\right)\Pi_{ij}(\bfk, t, t') \;=\; 0
\label{specprop}
\end{eqnarray} 
\end{theorem}
\begin{proof}
The velocity correlator in Fourier--space is
\begin{eqnarray}
\left<\tilde{v}_i(\bfK, \tau)\,\tilde{v}_j^*(\bfK', \tau')\right> &\;=\;& \int\,\mathrm{d}^3X\,\mathrm{d}^3X'\,
\exp{\left[\mathrm{i}\,\left(\bfK'\cendot\bfX' - \bfK\cendot\bfX\right)\right]}
\left<v_i(\bfX, \tau)\,v_j(\bfX', \tau')\right>\nonumber\\[3ex]
&\;=\;& \int\,\mathrm{d}^3X\,\mathrm{d}^3X'\,
\exp{\left[\mathrm{i}\,\left(\bfK'\cendot\bfX' - \bfK\cendot\bfX\right)\right]}\,\times\nonumber\\[1ex]
&&\qquad\times\,\left<v_i(\bfX + \bfX_c(\bfxi,\tau), \tau)\,v_j(\bfX' + \bfX_c(\bfxi,\tau'), \tau')\right> 
\label{fourcorr}
\end{eqnarray}
\noindent
where equation~(\ref{ginvacorr}) has been used. Using new dummy variables of integration, 
$\bfX \to \bfX - \bfX_c(\bfxi,\tau)$ and $\bfX' \to \bfX' - \bfX_c(\bfxi,\tau')$, we write
\begin{eqnarray}
\left<\tilde{v}_i(\bfK, \tau)\,\tilde{v}_j^*(\bfK', \tau')\right>  &\;=\;&
\exp{\left[\mathrm{i}\,\left(\bfK\cendot\bfX_c(\bfxi,\tau) - \bfK'\cendot\bfX_c(\bfxi,\tau')\right)\right]}\,\times\nonumber\\[1ex]
&&\;\;\times\,\int\,\mathrm{d}^3X\,\mathrm{d}^3X'\,\exp{\left[\mathrm{i}\,\left(\bfK'\cendot\bfX' - \bfK\cendot\bfX\right)\right]}
\left<v_i(\bfX, \tau)\,v_j(\bfX', \tau')\right>\nonumber\\[3ex]
&\;=\;& \exp{\left[\mathrm{i}\,\left(\bfK\cendot\bfX_c(\bfxi,\tau) - \bfK'\cendot\bfX_c(\bfxi,\tau')\right)\right]}\,\times\,
\left<\tilde{v}_i(\bfK, \tau)\,\tilde{v}_j^*(\bfK', \tau')\right>\nonumber\\
\label{fouriden}
\end{eqnarray}
\noindent
Comparing the left and right sides, we conclude that the phase,$\left[\bfK\cendot\bfX_c(\bfxi,\tau) - \bfK'\cendot\bfX_c(\bfxi,\tau')\right]$, 
must vanish. Substituting for $\bfX_c$ from equation~(\ref{orgvector}), the condition of zero phase implies that
\beq
(k_1 -k'_1)\xi_1 \;+\; (k_2 - k'_2)\xi_2 \;+\; (k_3 - k'_3)\xi_3 \;=\; 0
\label{phasezero}
\eeq
\noindent
where $\bfk\equiv (k_1, k_2, k_3)$ and $\bfk'\equiv (k'_1, k'_2, k'_3)$ are sheared wavevectors which are related to
 $\bfK$ and $\bfK'$ through the Fourier--space shearing transformation
\begin{eqnarray}
k_1 &\;=\;& K_1 + S\tau K_2\,,\qquad k_2 \;=\; K_2\,,\qquad k_3 \;=\; K_3\,,\qquad t \;=\; \tau\nonumber\\[2ex]
k'_1 &\;=\;& K'_1 + S\tau' K'_2\,,\qquad k'_2 \;=\; K'_2\,,\qquad k'_3 \;=\; K'_3\,,\qquad t' \;=\; \tau'
\label{fourshtr}
\end{eqnarray}
\noindent
Since equation~(\ref{phasezero}) must be valid for arbitrary $(\xi_1, \xi_2, \xi_3)$, we must have $\bfk=\bfk'$. 
In other words, the G--invariant Fourier--space velocity correlator must be of the general form stated in equation~(\ref{speca}). 
Moreover the listed properties of the velocity spectrum tensor, $\Pi_{ij}$, given in equations~(\ref{specprop}) follow from  the reality of 
$\bfv(\bfX, \tau)$, symmetry with respect to simultaneous interchange of $(i,j)$ and $(t,t')$, and incompressibility.
\end{proof}

\label{lastpage}

\end{document}